\documentclass[11pt]{article} 
\usepackage{times, amsmath, amsfonts, amsthm, physics, amssymb}

\usepackage[dvips]{graphicx,graphics}
\graphicspath{ {./graphs/} }
\usepackage{algorithm}
\usepackage{algpseudocode}
\usepackage{subcaption}
\usepackage{pdfpages}
\usepackage{blkarray}
\usepackage{bbm}
\usepackage{float}
\usepackage{bbm}
\usepackage[normalem]{ulem}
\usepackage{multirow}
\usepackage{gensymb}
\usepackage{textcomp}

\usepackage{color, colortbl}
\definecolor{Gray}{gray}{0.9}

\newtheorem{theorem}{Theorem}
\newtheorem{prop}{Proposition}

\DeclareMathSymbol{\sminus}{\mathbin}{AMSa}{"39}
\usepackage{url}
\usepackage{hyperref}
\hypersetup{colorlinks,%
citecolor=blue,%
filecolor=blue,%
linkcolor=red,%
urlcolor=blue,%
}

\DeclareMathOperator*{\argmin}{arg\,min}

\begin{document}

\title{Topologically Invariant Permutation Test}
\author{Sixtus Dakurah\\
Department of Statistics, University of Wisconsin-Madison\\
\textit{sdakurah@wisc.edu}}

\date{}
\maketitle

\begin{abstract}
Functional brain networks exhibit topological structures that reflect neural organization; however, statistical comparison of these networks is challenging for several reasons. This paper introduces a topologically invariant permutation test for detecting topological inequivalence. Under topological equivalence, topological features can be permuted separately between groups without distorting individual network structures. The test statistic uses $2$-Wasserstein distances on persistent diagrams, computed in closed form. To reduce variability in brain connectivities while preserving topology, heat kernel expansion on the Hodge Laplacian is applied with bandwidth $t$ controlling diffusion intensity. Theoretical results guarantee variance reduction through optimal Hilbert space projection. Simulations across diverse network topologies show superior performance compared to conventional two-sample tests and alternative metrics. Applied to resting-state fMRI data from the Multimodal Treatment of ADHD study, the method detects significant topological differences between cannabis users and non-users.\\

\end{abstract}

\section{Introduction}

The human brain operates as a complex network of interconnected regions, with neural activity continuously synchronized across spatially distributed areas. Understanding the organizational principles of brain networks is fundamental to neuroscience, as network topology influences cognitive function, develops through learning, and degrades in neurological disease. Functional magnetic resonance imaging (fMRI) has revolutionized our ability to measure these networks non-invasively by quantifying temporal correlations between regional brain activity patterns. These measurements yield functional connectivity networks, mathematical graphs where nodes represent brain regions and edges encode statistical dependencies between their activity time series \cite{bullmore2009complex,hagmann2008mapping}.

Traditional approaches to brain network analysis have primarily employed graph-theoretic measures such as clustering coefficients, path lengths, modularity, and hub identification. While these methods have yielded important insights into small-world organization and network efficiency \cite{newman2018networks}, they fundamentally treat networks as collections of pairwise relationships, potentially missing higher-order organizational features. The topology of a network, its intrinsic shape characterized by connected components, cycles, and voids, represents complementary information that may be critical for understanding brain function. A network's topology is invariant to continuous deformations, making it robust to noise and capturing global organizational principles that local graph metrics may overlook.

Persistent homology, a cornerstone of topological data analysis (TDA), provides mathematical tools for characterizing the multi-scale topological features of data \cite{edelsbrunner2008persistent}. Unlike traditional statistical methods that operate on point clouds or summary statistics, persistent homology tracks how topological features (connected components, loops, and higher-dimensional voids) emerge and disappear as data is examined at different scales \cite{dakurah2025robust,dakurah2024subsequence,dakurah2025maxtda,edelsbrunner2008persistent}. For brain networks, this scale-dependent perspective is particularly relevant: connectivity patterns exhibit hierarchical organization from local circuits to large-scale systems, and pathological changes may manifest at specific organizational scales \cite{dakurah2025discrete,dakurah2022modelling}. The power of persistent homology lies in its compression of complex network structure into persistence diagrams or barcodes. These representations encode the ``birth'' and ``death'' of topological features across a filtration, a nested sequence of networks obtained by thresholding edge weights.

While persistent homology characterizes global topology, analyzing signals defined on network edges such as functional connectivity strengths requires operators that respect the network's combinatorial structure. Traditional graph signal processing methods have been successfully applied to neuroimaging data \cite{hu2016matched,ortega2018graph}, enabling smoothing, denoising, and filtering of signals defined on graph vertices. However, these node-based approaches capture only pairwise interactions and cannot directly process edge-domain signals without auxiliary constructions. The Hodge Laplacian generalizes the graph Laplacian from nodes to higher-dimensional simplices, providing a natural framework for processing edge-domain signals \cite{barbarossa2020topological,schaub2018flow}. For brain connectivity networks, this enables smoothing and denoising operations that regularize connectivity patterns while preserving topological structure \cite{dakurah2025discrete,dakurahregistration}.

A common challenge in neuroimaging is determining whether observed differences between network populations reflect genuine biological variation or mere sampling variability. Standard two-sample tests, such as supremum tests on edge weights and conventional permutation methods, often fail to preserve network topological invariants during resampling, leading to inflated Type I error or reduced power \cite{chung2019rapid,nichols2002nonparametric}. Two complementary strategies address this limitation: heat kernel expansion and topologically invariant permutation. Heat kernel expansion via the Hodge Laplacian admits a spectral decomposition with controllable bandwidth, providing theoretically guaranteed variance reduction while maintaining signal fidelity within the span of leading eigenvectors. The expansion over the $i$-th simplex can be expressed through heat kernel convolution, yielding a spectral representation amenable to degree-$d$ expansion with Hilbert space optimality properties. This approach reduces variability in simplicial complexes, as first demonstrated in \cite{dakurah2025discrete}. Comparison of smoothed networks is then performed through a topologically invariant permutation procedure, wherein birth and death values from persistence diagrams are permuted independently while preserving each network’s intrinsic topology. This resampling scheme produces valid null distributions even for networks with differing global structures. Birth values correspond to the emergence of connected components ($0$-cycles) during filtration, while death values mark the destruction of $1$-cycles. Under the null hypothesis of topological equivalence, these distributions are exchangeable between groups, justifying their separate permutation.

The proposed test statistic combines Wasserstein distances on $0$D and $1$D persistent diagrams, forming a topological loss function that quantifies dissimilarity between network collections. The $2$-Wasserstein distance on persistent diagrams has a closed form for birth and death distributions, enabling efficient computation without optimal transport solvers. Under the null hypothesis of topological equivalence, the test statistic should be relatively small; large values indicate topological inequivalence. The resulting test procedure is sensitive to topological differences while controlling Type I error rates, as we demonstrate through extensive simulations.

This work makes several theoretical and methodological contributions to network neuroscience. First, we introduce a topologically invariant permutation test with rigorous theoretical justification for its validity. By permuting birth sets and death sets independently, the resampling scheme preserves the topological structure of individual networks while generating null distributions appropriate for testing topological equivalence. This addresses a fundamental limitation of conventional permutation tests that may distort network topology during resampling. Second, we demonstrate that heat kernel expansion on the Hodge Laplacian reduces variance while preserving topology, with formal guarantees through Hilbert space projection (Theorem~\ref{thm:theorem-hilbert}). The heat kernel expansion provides optimal approximation in the subspace spanned by low-frequency eigenvectors, with the bandwidth parameter $t$ and expansion degree $d$ offering intuitive control over the smoothing process. We prove that heat kernel smoothing reduces the variability of functional measurements across networks at each edge, enabling improved signal-to-noise ratio without destroying topological features of interest. Third, we provide extensive simulation evidence that our approach outperforms existing methods across diverse network topologies and connectivity patterns. Four comprehensive simulation studies validate: (i) the topologically invariant inference procedure successfully preserves topology while distinguishing networks with different structures (circles, lemniscates, quadrifoliums); (ii) superior performance compared to baseline two-sample $t$-tests, especially for networks with low intramodule connectivity; (iii) enhanced performance gains from heat kernel smoothing; and (iv) advantages over alternative loss functions including $\ell_1-$, $\ell_2-$, $\ell_\infty-$norms, bottleneck distance, and Gromov-Hausdorff distance. The practical significance of these contributions is demonstrated through application to real fMRI data investigating cannabis use effects on brain networks.

The remainder of this paper is organized as follows. Section~\ref{sec:methods} develops the mathematical framework, including simplicial homology, graph filtration, Hodge Laplacians, heat kernel smoothing, and the topologically invariant inference procedure. We present the theoretical foundation for boundary operators, prove key results on Betti numbers and variance reduction, and formalize the permutation test procedure with its topological loss function. Section~\ref{sec:app} presents four comprehensive simulation studies validating the method's performance across diverse scenarios, followed by application to the MTA resting-state fMRI dataset with detailed statistical results. Section~\ref{sec:disc} discusses theoretical implications, methodological advantages for neuroimaging, biological interpretation of findings, limitations, and future directions for topological network analysis in neuroscience.

\section{Methods}
\label{sec:methods}

\subsection{Simplical Complex and Simplicial k-chains} 
A k-simplex is a k-dimensional polytope whose convex hull is made up of its k+1 nodes. A simplicial complex $\mathcal{K}$ is a set of simplices such that for any $\sigma_1, \sigma_2 \in \mathcal{K}$, $\sigma_1 \cap \sigma_2$ is a face of both simplices; and a face of any simplex $\sigma \in \mathcal{K}$ is also a simplex in $\mathcal{K}$ \cite{wang2012basic}.
For some $r_i \in \mathbb{R}$, the the finite sum over the simplices $\sum_{i=1}^N r_i \sigma_i$, $\sigma_i \in \mathcal{K}$ is termed the simplicial k-chain. The set of simplicial k-chains with addition over $\mathbb{R}$ is used to construct boundary maps \cite{topaz2015topological, wang2012basic}.\\

\subsubsection{Boundary Map of Simplicial k-chains} Take any two set of simplicial k-chains $\mathcal{K}_k$ and $\mathcal{K}_{k-1}$, the boundary map 

\begin{equation}
    \partial_k(\sigma): \mathcal{K}_k \longrightarrow \mathcal{K}_{k-1} 
    \label{eqn:init_boundary_map}
\end{equation}
\noindent
for each k-simplex $\sigma = (v_0, ..., v_k)$ is given as: 
\begin{equation}
    \partial_k(\sigma) = \sum_{i=1}^N(-1)^i(v_0, ..., \hat{v_i}, ..., v_k)
    \label{eqn:boundary_map}
\end{equation}
where $(v_0, ..., \hat{v_i}, ..., v_k)$ is the k-1 face of $\sigma$ obtained by deleting the $v_i$ node.\\

\subsubsection{Homology of a Simplicial Complex} Two important components of the boundary map (\ref{eqn:boundary_map}) are its kernel denoted as $\mathcal{Z}_k = ker(\partial_k)$, and its image denoted as $\mathcal{B}_{k} = img(\partial_{k+1})$. These are subspaces of $\mathcal{K}_k$. The elements of $\mathcal{Z}_k$ are known as $k$-cycles and that of $\mathcal{B}_{k}$ are known as $k$-boundaries \cite{hatcher2002algebraic}. In essence $\mathcal{Z}_k$ is the subspace of  $\mathcal{K}_k$ consisting of $k$-chains that are also $k$-cycles and $\mathcal{B}_{k}$ is a subspace of $\mathcal{K}_k$ consisting of $k$-cycles that are also $k$-boundaries.\\
The set quotient $\mathcal{H}_k = \mathcal{Z}_k/\mathcal{B}_{k}$ is termed the $kth$ homology-module and its an $\mathbb{R}$-module. It can be shown that $\mathcal{B}_{k} \subseteq \mathcal{Z}_k$ \cite{hatcher2002algebraic, topaz2015topological}. Hence the difference, can intuitively be thought of as the failure of k-cycles in $\mathcal{K}$ to bound $(k+1)$-simplices, leading to the concept of "holes". A measure of this number of $k$-dimensional holes is termed the Betti number denoted $\beta_{k} = rank(\mathcal{H}_k)$.\\
If we can equip $\mathcal{K}_k$ with the standard $\mathcal{L}^2$ inner product, which will result in the Hilbert space denoted $\mathbb{H}(\mathcal{K}_k)$, and we define a basis for each of these operators, we can abstract the chain operations in terms of matrix multiplication which will allow us to have a matrix representation of the boundary operators.\\

 \subsubsection{Homology of Graphs} 
 
 An unoriented graph can be viewed as a 1-dimensional simplicial complex. Consider the following characterization of a graph $\mathcal{X} = (V, E, w)$. Define $\mathcal{V}$ and $\mathcal{E}$ to be the respective node and edge space  of the graph. The formal sum over the node space
 
\begin{equation}
    \sum_{v\in \mathcal{V}} z_vv, \hspace{1cm} z_v \in\mathbb{Z}
    \label{eqn:hmg_ochains}
\end{equation}
 \noindent
 where $z_v = 0$ for all but a finite $v\in \mathcal{V}$. Note that this sum is a combination of 0-simplices (nodes) and is a 0-chain denoted $\mathcal{K}_0(\mathcal{X})$. Similarly, a formal sum over the directed edge space (combination of 1-simplices) will give the 1-chain denoted $\mathcal{K}_1(\mathcal{X})$. Similar to the construction of homology groups for simplicial complexes, we can define the boundary map
 
 $$\partial_1: \mathcal{K}_1(\mathcal{X}) \longrightarrow \mathcal{K}_{0}(\mathcal{X})$$
 
 Since $img(\partial_2) = 0$, we have that the first homology module $\mathcal{H}_1(\mathcal{X}) = \ker(\partial_1)$. If $\mathcal{X}$ is finite, then $\mathcal{H}_1(\mathcal{X})$ has finite rank, given by the first Betti number $\beta_1$. Elements of $\mathcal{H}_1(\mathcal{X})$ are termed 1-cycles.  
\begin{prop}
    The number of loops (or 1-cycles) in a complete graph $\mathcal{X}$ with $p$ nodes is $\frac{1}{2}p^2 -\frac{3}{2}p + 1$.
\end{prop}

\begin{proof}
Let the complete graph $\mathcal{X}$ has the form $\mathcal{X} = (V, E, w)$. Since $\mathcal{X}$ is complete, it's also connected with 1 connected component. The $0$-$th$ Betti number $\beta_0$ is the number of connected components of $\mathcal{X}$. From the discussion in homology of graphs, the number of 1-cycles will be given by the $1st$ Betti number $\beta_1$.
The Euler characteristic function of $\mathcal{X}$ is given by $\chi(\mathcal{X}) = |V| - |E| = \beta_0 - \beta_1$ \cite{adler2010persistent, chung2019exact}.\\
\noindent
The graph $\mathcal{X}$ is complete $\implies$ $|V| = p$, $|E| = \frac{|V|(|V|-1)}{2} = \frac{p(p-1)}{2}$. Hence 

    $$B_1 = 1 - p + \frac{p(p-1)}{2} = \frac{1}{2}p^2 -\frac{3}{2}p + 1 = \mathcal{O}(p^2)$$
\end{proof}

\subsubsection{Graph Filtration} Given ordered thresholds (filtration values) ${{\epsilon}_0} < ... < {{\epsilon}_k}$, a filtration of the graph $\mathcal{X}$ is a collection of a sequence of nested binary networks \cite{chung2019exact}

\begin{equation}
    \mathcal{X}_{{\epsilon}_0} \supset ... \supset \mathcal{X}_{{\epsilon}_k}
    \label{eqn:binary_network}
\end{equation}
Intuitively, since $\mathcal{X}_0$ is complete, the number of connected components is therefore 1, as we increase threshold $\epsilon$, more edges are disconnected, increasing the number of connected components $(\beta_0)$, and decreasing the number of cycles $(\beta_1)$. More formally, it has been shown that in a graph $\mathcal{X}$, Betti numbers $\beta_0$ and $\beta_1$ are monotone over graph filtration on edge weights \cite{chung2019exact}.\\

\subsubsection{Birth-death Decomposition}
The filtration (\ref{eqn:binary_network}) allows us to track the birth of connected components ($0$-cycles) and the death of $1$-cycles over a given span. Persistence modules, which uniquely decomposes into persistence intervals \cite{zomorodian2005computing} is a tool for tracking this evolution. The persistence of a connected component or 1-cycle that appears at filtration value $b_i$ and disappears at filtration value $d_i$ is represented by the interval $\left[b_i,  d_i\right]$. A finite collection of the persistence intervals can be summarized in the form of a \textit{barcode}.
 Note that when a connected component is born, it never dies, hence it's lifespan is $[b_i, \infty)$. For a graph $\mathcal{X}$ with p nodes, ignoring the $\infty$ death value, the total number of birth values of the complete graph is given by $p-1$. We can represent a collection of the birth values $B(\mathcal{X})$ by the 0D barcode 
$$ B(\mathcal{X}) = b_1  < b_2 < \dots < b_{p-1} $$ 
 Similarly, all the 1-cycles are born when the graph is first formed and it's lifespan is represented as $(-\infty, d_i]$. During the span of the filtration, when an edge is deleted, a 0-cycle is formed or a 1-cycle dies. Both events can not occur at the same time \cite{songdechakraiwut2020topological}.  At the point $d_i$, we take the death value to be the weight of the edge in $\mathcal{X}$ that has been deleted. Since $\mathcal{X}$ is complete, it has $\frac{p(p-1)}{2}$ unique edge weights, hence if we ignore the $-\infty$ birth value, the number of death values is given by\\ $q = \frac{p(p-1)}{2} - (p-1) = \frac{(p-1)(p-2)}{2}$. The death set $D(\mathcal{X})$ can be represented by the 1D barcode 
 $$ D(\mathcal{X}) = d_1 < d_2 < \dots < d_{q} $$
Denote the collection of edge weights of $\mathcal{X}$ by $\mathcal{W}$, then the sets $B(\mathcal{X})$ and $D(\mathcal{X})$ partition $\mathcal{W}$ such that $B(\mathcal{X}) \cup D(\mathcal{X}) = \mathcal{W}$ and $B(\mathcal{X}) \cap D(\mathcal{X}) = \emptyset$.\\

\subsection{Hodge Laplacian} 

Let $\mathbb{B}_k$ and $\mathbb{B}_{k+1}$ be the matrix representation of the boundary operators $\partial_k$ and $\partial_{k+1}$ respectively, the Hodge Laplacian is now defined as follows.
Denote by $\mathcal{L}_k$ the $k$-$th$ Hodge Laplacian. Then we have that 

\begin{equation}
    \mathcal{L}_k = \partial_{k+1}\partial_{k+1}^* + \partial_{k}^*\partial_{k}
    \label{eqn:hodge1}
\end{equation}
which is equivalent to 

\begin{equation}
    \mathcal{L}_k = \mathbb{B}_{k+1}\mathbb{B}_{k+1}^T + \mathbb{B}_{k}^T\mathbb{B}_{k}
    \label{eqn:hodge2}
\end{equation}
\noindent
in the matrix form, which is mainly used for actual computation. The boundary matrix $\mathbb{B}_k$ has the following entries 
\begin{equation}
	(\mathbb{B}_k)_{ij} =
	\begin{cases}
		1, 	& \text{if } \sigma^i_{k-1}  \subset \sigma^j_{k}  ~~\text{and}~~ \sigma^i_{k-1} \sim \sigma^j_{k}\\
		-1, & \text{if } \sigma^i_{k-1}  \subset \sigma^j_{k}  ~~\text{and}~~ \sigma^i_{k-1} \nsim \sigma^j_{k}\\
		0,  & \text{if } \sigma^i_{k-1}  \not \subset \sigma^j_{k}
	\end{cases},
\label{Eq:boundarymatrix}
\end{equation}
It's straightforward to check that, $\mathcal{L}_k$ is a $|\mathcal{K}_k|\times |\mathcal{K}_k|$ matrix, where $|\mathcal{K}_k|$ is the cardinality of k-simplices in $\mathcal{K}_k$. An illustration of the computation of the boundary matrices is shown in Figure \ref{fig:schematic-boundary-operators}. For example, take $k = 0$, then $\mathcal{L}_0 = \mathbb{B}_1\mathbb{B}_1^T$ is the usual graph Laplacian, computed as the difference between the degree matrix and the adjacency matrix. Observe that $\partial_1$ maps from edges to nodes and its matrix representation $\mathbb{B}_1$ relates edges to nodes, the incidence matrix in graph theory literature. When $k = 1$ and the complex is a 1-skeleton, we have that $\partial_2 = 0$ since there are no 2-simplices. Then $\mathcal{L}_1 = \mathbb{B}^T_1\mathbb{B}_1$.\\

\begin{figure}[ht]
 \centering
 \includegraphics[width =\textwidth]{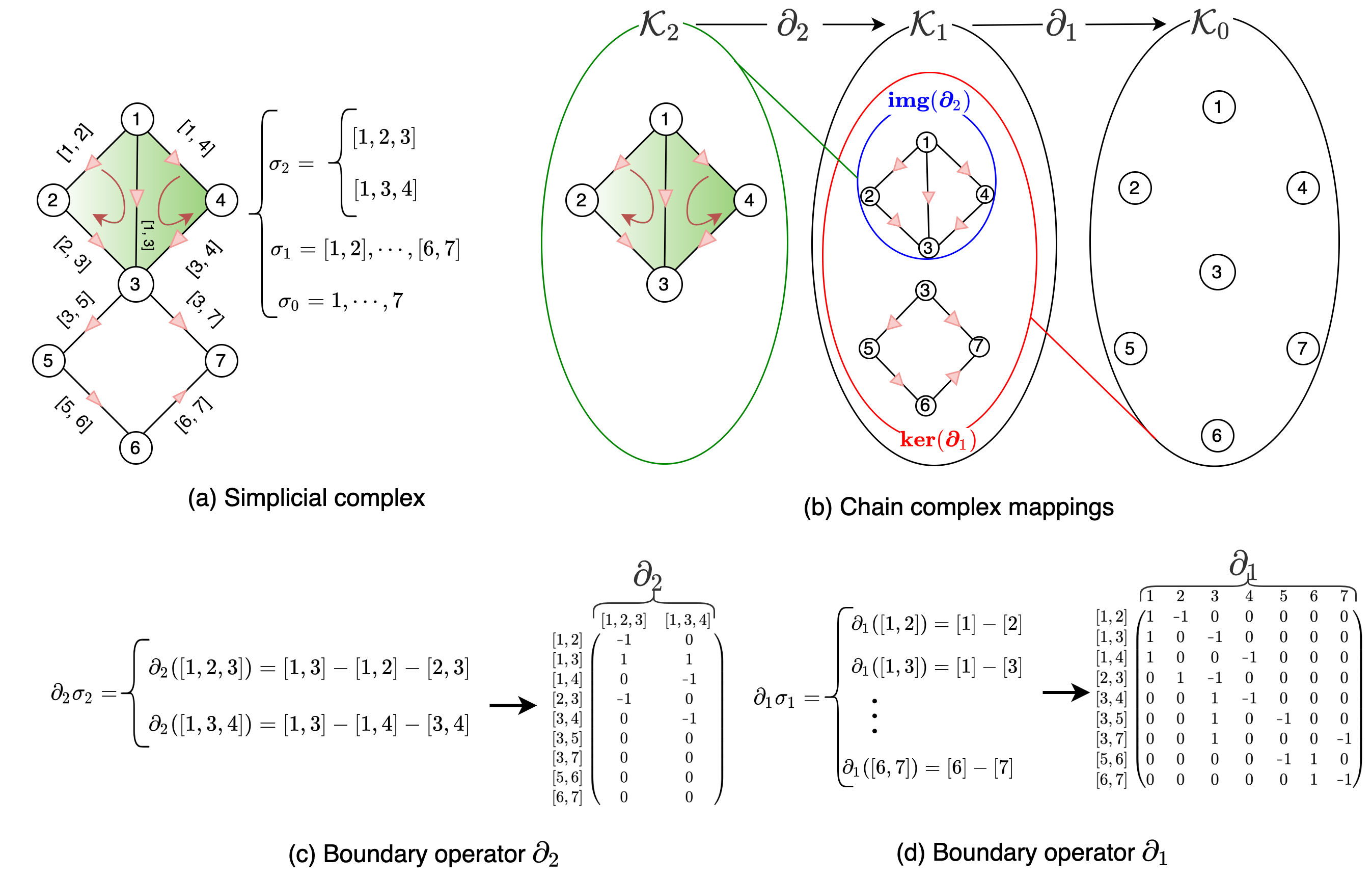}
 \caption{(a) A 2-dimensional simplicial complex. (b) A chain complex mapping with boundary operators that map higher dimensional simplices to lower dimensional simplices. (c) The boundary operator $\partial_2$ acting on 2-simplices $\sigma_2$ to give the matrix form ($\mathbb{B}_2$). (c) The boundary operator $\partial_1$ acting on 1-simplices $\sigma_1$ to give the matrix form ($\mathbb{B}_1$).}
 \label{fig:schematic-boundary-operators}
\end{figure}

Given a collection of $k$-simplices $\mathcal{S}_1, \cdots, \mathcal{S}_n$, these can be combined according to a chosen criteria to form a template simplicial complex denoted $\bar{\mathcal{S}}_k$. The $k$-th Hodge Laplacian $\mathcal{L}_k$ can then be constructed from this template $\bar{\mathcal{S}}_k$. The eigendecomposition of $\mathcal{L}_k$ can be obtained as
\begin{equation} 
\mathcal{L}_k \mathbf{u}_k =  \lambda_k \mathbf{u}_k \label{eqn:eigendecomposition-gen}
\end{equation}
The eigenvectors form an orthonormal basis satisfying $\mathbf{u}_i^{\top} \mathbf{u}_j  = \delta_{ij}$ with Kroneker delta $\delta_{ij}$.

\subsection{Heat kernel expansion through Hodge Laplacian }
Brain images are often denoised to increase the signal-to-noise ratio (SNR) and enhance statistical sensitivity. Denoising induces many nice statistical properties such as variance reduction \cite{chung2017heat}. Brain connectivity matrices are noisy so it is necessary to denoise the matrices as well. 
Heat kernel smoothing can be used to smooth over 1-simplices \cite{anand2021hodge}. The diffusion equation over the $i$-$th$ simplex can be defined as
\begin{equation}
    \frac{\partial}{\partial t} \mathbf{f}_q(t, i)  = - \mathcal{L}_k\mathbf{f}_q(t, i)
    \label{eqn:diffusion-eqn}
\end{equation}
The solution to this equation can be obtained via the heat kernel convolution. Using the eigendecomposition (\ref{eqn:eigendecomposition-gen}), we can state the heat kernel as follows:
\begin{equation}
    K_t(i_1, i_2) = \sum_{j = 0}^m e^{-\lambda_j t}\mathbf{u}_j(i_1)\mathbf{u}_j(i_2)
    \label{eqn:heat-kernel}
\end{equation}
where $t$ is the diffusion bandwidth. Then we have the following spectral representation of the heat kernel smoothing:
\begin{equation}
    K_t*f_q(i) = \sum_{j=0}^m e^{-\lambda_j t}{b}_j \mathbf{u}_j(i)
\end{equation}
Here, ${b}_j = \langle \mathbf{f}_q, \mathbf{u}_j \rangle$ are the Fourier coefficients.
The degree $d$ expansion of the above has a Hilbert space interpretation given in the theorem below.
\begin{theorem}
Let the subspace 

$$\mathcal{U} = \left\{ \sum_{j=0}^d \alpha_j \mathbf{u}_j: \alpha_j \in \mathbb{R}  \right\} \subset \mathcal{L}^2$$
which is spanned by up to the $d$-$th$ degree basis vectors. Let $\mathbf{f}_q$ be the solution to (\ref{eqn:diffusion-eqn}). Then the closest function $\mathbf{u}$ in the subspace $\mathcal{U}$ to $\mathbf{f}_q$ is given by
\begin{equation}
     \argmin_{\mathbf{u}}\|\mathbf{f}_q - \mathbf{u}\|_2^2 = \sum_{j=0}^d e^{-\lambda_j t} b_j \mathbf{u}_j(i)
     \label{eqn:argmin-eq}
\end{equation}
\label{thm:theorem-hilbert}
\end{theorem}
The above theorem provides a new framework for performing diffusion smoothing with a series expansion on the Hodge Laplacian. The level of smoothing is then controlled by the bandwidth $t$ over some expansion degree $d$. The parameters $b_j$'s in (\ref{eqn:argmin-eq}) will be estimated in the least squares fashion.\\
If we denote $\mathbb{V}_i \mathbf{f}_q$ to be the variance of functional measurement $\mathbf{f}_q$ across networks at the $i$-th edge, the following results holds.
\begin{prop}
Heat kernel smoothing reduces variability, i.e., 
\begin{equation}
    \mathbb{V}_i \mathbf{f}_q(t) \leq \mathbb{V}_i \mathbf{f}_q
\end{equation}
\end{prop}

\subsection{Topologically invariant inference}

We investigate  whether two groups of networks are topologically equivalent. The 0D and 1D barcodes are topological features that completely characterizes a network. The topological similarity or dissimilarity of networks can be assessed based on these barcodes. 
 To this end, we propose a permutation test on the persistent diagrams (PD) of the barcodes that employs a topologically invariant resampling technique. We propose the following resampling procedure which is topologically invariant.\\
 
 Given a collection of networks $\mathcal{N} = \{ \mathcal{X}^1, \mathcal{X}^2, \cdots, \mathcal{X}^n \}$, we compute the birth-death decomposition as outlined in the previous section and obtain the respective birth $B(\mathcal{N})$ and death $D(\mathcal{N})$ sets. For two non-overlapping categorization of $\mathcal{N}$, denote the birth and death sets in the two groups as
$$
\boldsymbol{\xi}_{1} = \{ B(\mathcal{X}^{1}), \cdots, B(\mathcal{X}^{n_{1}}) \} \hspace{0.5cm} \boldsymbol{\zeta}_{1} = \{ D(\mathcal{X}^{1}), \cdots, D(\mathcal{X}^{n_{1}})  \}
$$
$$\boldsymbol{\xi}_{2} = \{ B(\mathcal{X}^{1}), \cdots, B(\mathcal{X}^{n_{2}}) \} \hspace{0.5cm} \boldsymbol{\zeta}_{2} = \{ D(\mathcal{X}^{1}), \cdots, D(\mathcal{X}^{n_{2} })  \}
$$
such that $n_1 + n_2 = n$.
To create the invariant resampling scheme, we permute $\boldsymbol{\xi}_{1}$ with $\boldsymbol{\xi}_{2}$, and similarly, permute $\boldsymbol{\zeta}_{1}$ with $\boldsymbol{\zeta}_{2}$. Denote these permutations as $\boldsymbol{\tau}(\boldsymbol{\xi}_{1}, \boldsymbol{\xi}_{2})$ and $\boldsymbol{\tau}(\boldsymbol{\zeta}_{1}, \boldsymbol{\zeta}_{2})$. We can now build the permutation test to test for topological equivalence using this resampling scheme.\\ 

More formally, take two network collections $\mathcal{N}_1 = \{ \mathcal{X}^{1}, \dots, \mathcal{X}^{m} \}$ and $\mathcal{N}_2 = \{ \mathcal{Y}^{1}, \dots, \mathcal{Y}^{n} \}$, the null hypothesis statement is given as:
\begin{equation}
    H_0: \mathbb{P}_1(\boldsymbol{\xi}_{1}) \equiv \mathbb{P}_1(\boldsymbol{\xi}_{2}), \hspace{0.5cm} \mathbb{P}_2(\boldsymbol{\zeta}_{1}) \equiv  \mathbb{P}_2(\boldsymbol{\zeta}_{2}) \label{eqn:top-hyp}
\end{equation}
where $\mathbb{P}_{1}(.)$ and $\mathbb{P}_{2}(.)$  is the persistent diagram on the 0D and 1D barcodes. Essentially, the 0D persistent diagrams between the two groups are identical and the 1D persistent diagram between the two groups are identical. Under this assumption, we can permute the birth sets between the two groups and the death set between the two groups separately.
Then we define topological loss $\mathcal{L}( \mathcal{X}, \mathcal{Y})$ as the sum of \textit{2-Wasserstein} distances of 0D and 1D persistent diagrams. Let $D_{W}(.)$ denote the Wasserstein distance on the 0D and ID persistent diagrams. The topological loss can be stated as
$$
\mathcal{L}( \mathcal{X}, \mathcal{Y}) = D_{W_0}^2(\mathbb{P}_{1}(\boldsymbol{\xi}_{1}), \mathbb{P}_{1}(\boldsymbol{\xi}_{2})) + D_{W_1}^2( \mathbb{P}_{2}(\boldsymbol{\zeta}_{1}), \mathbb{P}_{2}(\boldsymbol{\zeta}_{2}) ) 
$$
It can be shown that the Wasserstein distance on the respective 0D and 1D persistent diagrams has the form
\begin{equation}
     D_{W_0}^2\left[\mathbb{P}_{1}(\boldsymbol{\xi}_{1}), \mathbb{P}_{1}(\boldsymbol{\xi}_{2})\right] = \sum_{k=0}^{q_0} \left[b_{(k)}^1 - b_{(k)}^2\right]^2
\end{equation}
\begin{equation}
     D_{W_1}^2\left[\mathbb{P}_{1}(\boldsymbol{\zeta}_{1}), \mathbb{P}_{1}(\boldsymbol{\zeta}_{2})\right] = \sum_{k=0}^{q_1} \left[d_{(k)}^1 - d_{(k)}^2\right]^2
\end{equation}
where $b_{(k)}^1$ and $b_{(k)}^2$ are the $k$-$th$ smallest birth values in $\mathbb{P}_{1}(\boldsymbol{\xi}_{1})$ and $\mathbb{P}_{1}(\boldsymbol{\xi}_{2})$ respectively, and similarly, $d_{(k)}^1$ and $d_{(k)}^2$ are the $k$-$th$ smallest death values in $\mathbb{P}_{1}(\boldsymbol{\zeta}_{1})$ and $\mathbb{P}_{1}(\boldsymbol{\zeta}_{2})$ respectively. The test statistic is then constructed to reflect the expected small Wasserstein distance for networks of the same topology and large distance for networks of different topologies. Let $\mathcal{L}_W$ denote the total Wasserstein distance for networks within the same group 
\begin{equation}
\mathcal{L}_W = \frac{\sum_{\mathcal{X} \in \mathcal{N}_1, \mathcal{Y} \in \mathcal{N}_1}  \mathcal{L} (\mathcal{X}, \mathcal{Y}) + \sum_{\mathcal{X} \in \mathcal{N}_2, \mathcal{Y} \in \mathcal{N}_2}  \mathcal{L} (\mathcal{X}, \mathcal{Y})}{\binom{n_1}{2} + \binom{n_2}{2}}
\end{equation}
and $\mathcal{L}_B$, the between-group total Wasserstein distance
\begin{equation}
\mathcal{L}_B= \frac{\sum_{\mathcal{X} \in \mathcal{N}_1, \mathcal{Y} \in \mathcal{N}_2}  \mathcal{L} (\mathcal{X}, \mathcal{Y})}{n_1\times n_2}
\end{equation}
The test statistic is given as the ratio of these two quantities:
\begin{equation}
\mathcal{T}(\mathcal{N}_1, \mathcal{N}_2) = \mathcal{L}_B/\mathcal{L}_W
\label{eqn:total-top-loss-test-statistic}
\end{equation}
Under the null hypothesis of topological equivalence, the test statistic should be relatively small. A large test statistic indicates topological in-equivalence. The inference is done through the permutation test.

\section{Validation}
To validate the proposed methods, we perform a series of simulations to assess the performance of the topological invariant resampling scheme in preserving topological information. The test procedure build on this resampling scheme is also assessed for its performance to distinguish networks of different topologies. We also studied the variance reduction properties of the heat kernel smoothing.

\subsection{Study I: Topologically invariant inference preserves topology}

\begin{figure*}[ht]
     \centering
     \begin{subfigure}[t]{0.328\linewidth}
         \centering
         \includegraphics[width=\textwidth]{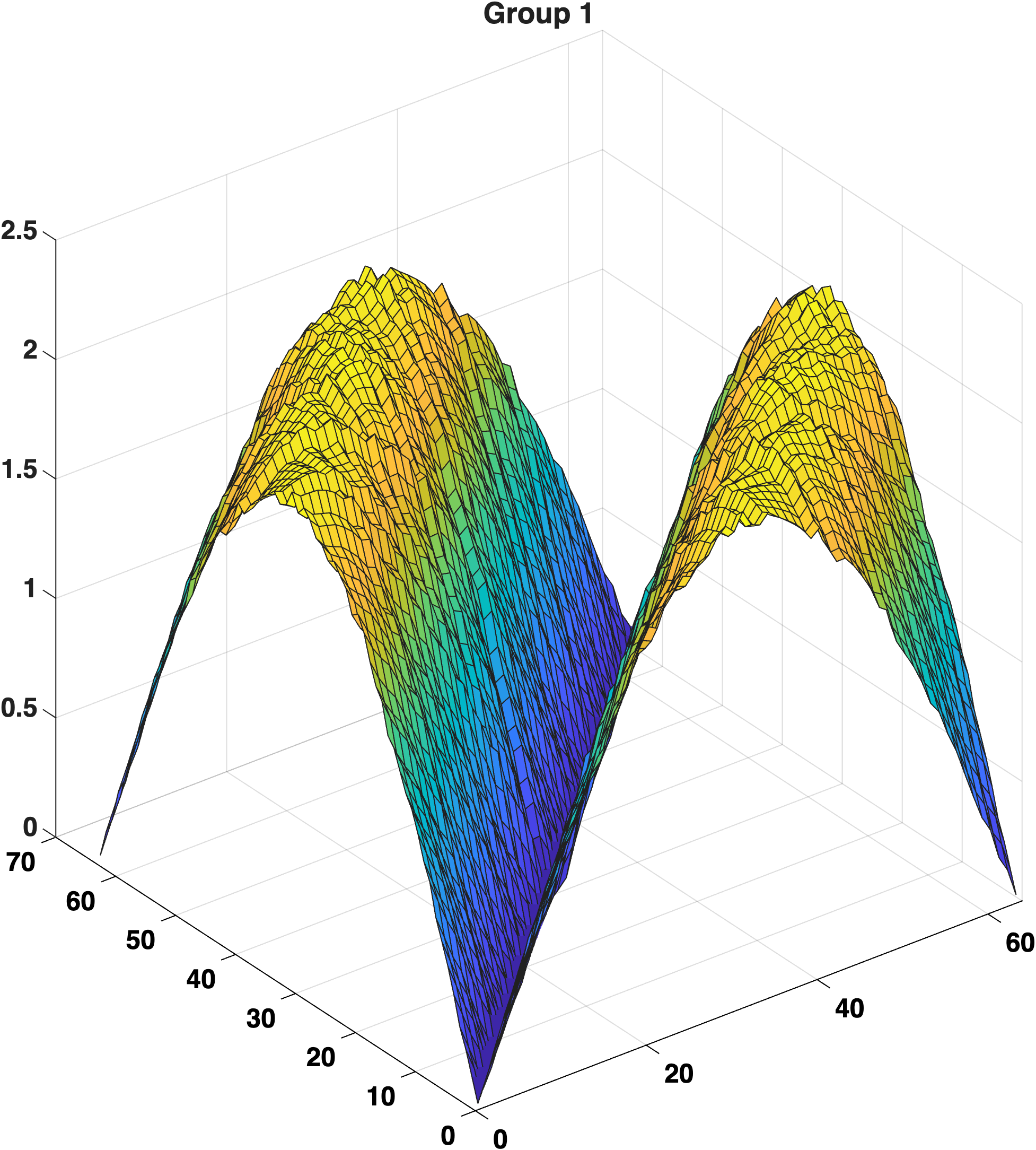}
     \end{subfigure}
     \hfill
     \begin{subfigure}[t]{0.328\linewidth}
         \centering
         \includegraphics[width=\textwidth]{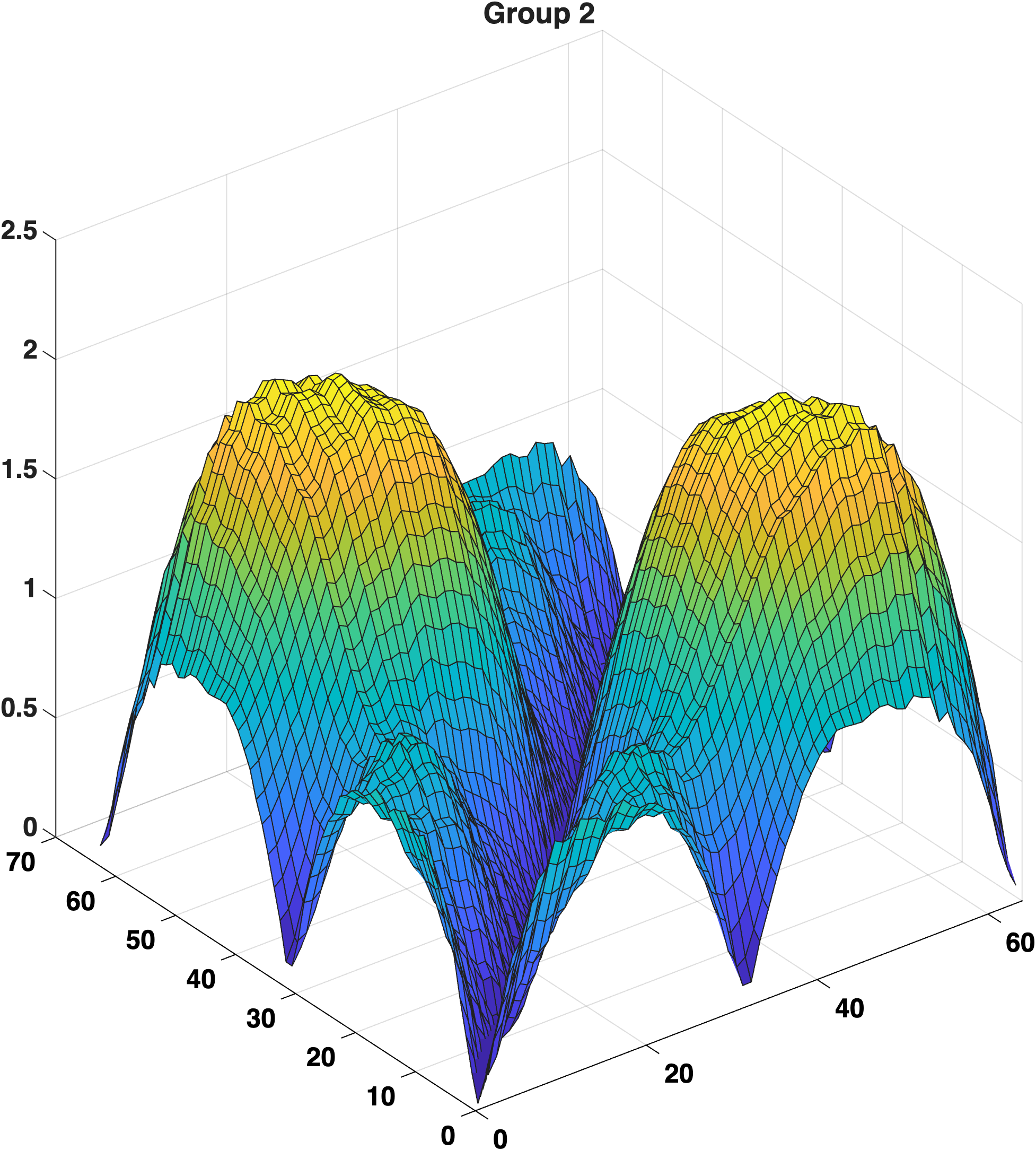}
     \end{subfigure}
     \hfill
     \begin{subfigure}[t]{0.328\linewidth}
         \centering
         \includegraphics[width=\textwidth]{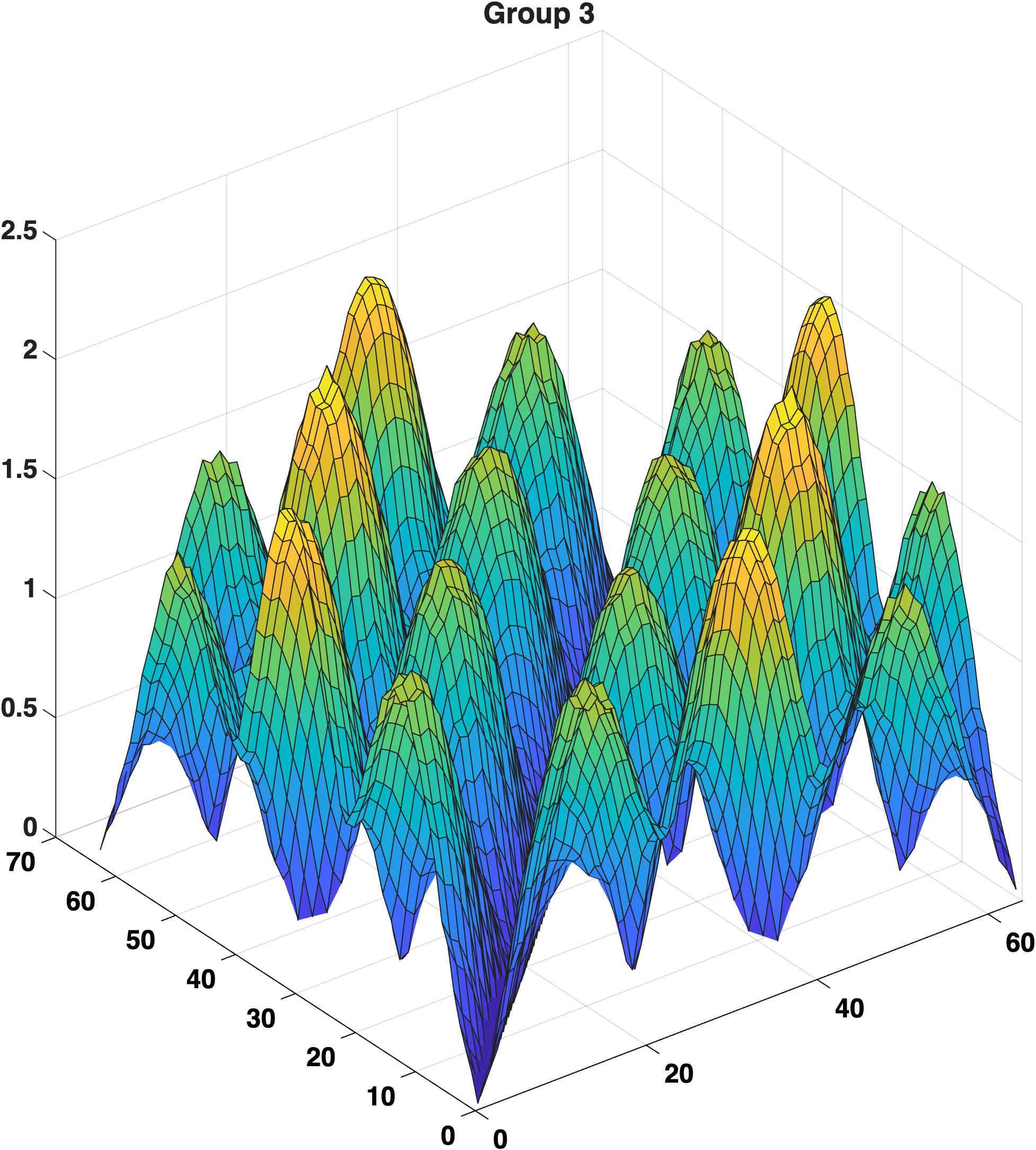}
     \end{subfigure}
    \caption{The three network groups: circle, lemniscate,
quadrifolium used in the study. The Gaussian noise N (0, 0.02) is added to the coordinates of curves.}
    \label{fig:loops}
\end{figure*}

The purpose of this simulation study is to validate that the proposed topologically invariant resampling scheme indeed preserves topology. The three group of networks with different topologies were generated \cite{anand2021hodge} (Figure \ref{fig:loops}). The first group comprises of circles, the second group has shapes of leminiscate, and the third group has shapes of the quadrifolium. The 64 points along these curves were perturbed with random noise $N(0, 0.02)$. We generated from 5 to 100 networks in each group to see how the number of samples affect the performance.

We then employed the proposed topologically invariant resampling scheme with topological loss between groups in Equation~\eqref{eqn:total-top-loss-test-statistic}. When comparing networks of the same topology, the resampling scheme will naturally preserve the topological invariants between the groups and the test procedure should confirm their topological equivalence. For networks with different topologies, the invariant inference procedure facilitates the exchange of corresponding topological invariants between the two groups without distorting their topologies. Within each group, it is expected topology are similar/identical. The test procedure will therefore indicate the topological inequivalence of the two groups. To estimate the p-value in each group comparisons, we used 100,000 permutations. The p-values reported are an average of 50 independent simulations.

\begin{table}[ht]
\caption{The performance results on differentiating networks of different topologies using the topological loss. Different number of networks from 5 to 100 were generated for each group.}
\centering
\begin{tabular}{c|c|c|c|c}
Group & 5 networks & 25 networks & 50 networks & 100 networks\\\hline
1 vs. 2 & 0.0000 (0.0000) & 0.0000 (0.0000) & 0.0000 (0.0000) & 0.0000 (0.0000) \\
1 vs. 4 & 0.0000 (0.0000) & 0.0000 (0.0000) & 0.0000 (0.0000) & 0.0000 (0.0000) \\
2 vs. 4 & 0.0000 (0.0000) & 0.0000 (0.0000) & 0.0000 (0.0000) & 0.0000 (0.0000) \\
\hline
1 vs. 1 & 0.0677 (0.0009) & 0.0399 (0.0006) & 0.1853 (0.0014) & 0.0265 (0.0005) \\
2 vs. 2 & 0.0000 (0.0000) & 0.0498 (0.0006) & 0.0303 (0.0005) & 0.0002 (0.0001) \\
4 vs. 4 & 0.0894 (0.0008) & 0.1007 (0.0009) & 0.0952 (0.0009) & 0.0795 (0.0009) \\
\hline
\end{tabular}
\label{tab:SimulationStudy1-tii}
\end{table}

Table~\ref{tab:SimulationStudy1-tii} summarizes the results for the topologically invariant test procedure. As expected, p-values are  close to 0 when there are topological differences and they are large when there is no topological difference. Further, as the sample size increases, the performance tend to increases. The study demonstrates that the topologically invariant resampling scheme indeed preserves topology and  can accurately distinguish networks of differing topologies with the exception of a few similar topological comparison that gave false positives which can likely be attributed to the random perturbation.

\subsection{Study II: Comparison against baseline two-sample t-test}
We compared the performance of the proposed {\em topologically invariant inference} procedure against the two-sample t-test often employed in brain network analysis\cite{alberton2020multiple,chung2019rapid,lindquist2015zen,nichols2002nonparametric}.
Unlike Study I, we use weighted random modular networks that exhibit more of 0D topology. The process for generating the random modular networks and assigning random weights to generated edges was adapted from \cite{songdechakraiwut2020topological}. We generated a random modular network with $p$ nodes and $K$ modules (Figure \ref{fig:modular-network}). The nodes were distributed evenly among the modules. The weight of any edge incident by two nodes in the same module was generated from $N(1, 0.25)$ with probability $\pi$ or $N(0, 0.25)$ with probability $1-\pi$. Weights of edges incident by two nodes in different modules were generated from $N(0, 0.25)$ with probability $\pi$ or $N(1, 0.25)$ with probability $1-\pi$. Negative edge weights were set to zero.  

\begin{figure}[ht]
 \centering
 \includegraphics[width =\textwidth]{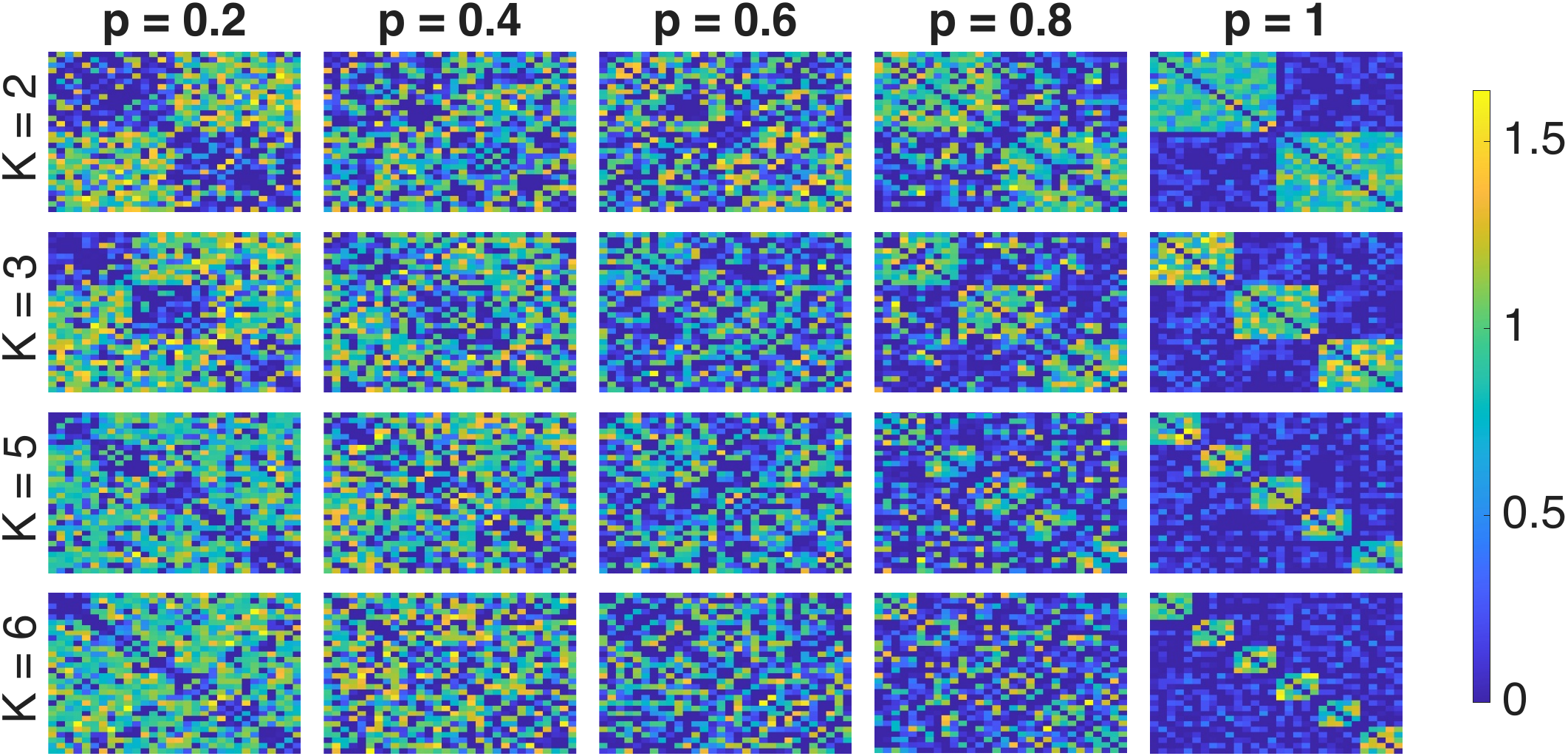}
 \caption{Study II simulation examples. The network modular structure varies as intramodule connectivity probability $\pi$ changes with the number of network modules $K$. The modular structure becomes pronounced as $\pi$ increases. The colorbar displays the edge weights. Edge weights within the same module were drawn from $N(1, 0.25)$ with probability $\pi$ or Gaussian noise $N(0, 0.25)$ with probability $1-\pi$. Analogously, edge weights connecting nodes between different modules have probability $1-\pi$ of being $N(1, 0.25)$ and probability $\pi$ of being $N(0, 0.25)$.}
 \label{fig:modular-network}
\end{figure}

{\bf Topological difference.} Two groups of networks with different modularity structures are compared, i.e, 2 vs. 5, 3 vs. 6. These groups were compared for various intramodule connectivity probabilities. The topology of each group will be different and we expect the {\em topologically invariant inference} procedure to detect the network difference. The results of these comparisons are summarized in Table \ref{tab:SimulationStudy2-diff}. The {\em topologically invariant inference} procedure outperforms the baseline method when distinguishing networks of different topologies, especially with low intramodule connectivity probability. The baseline method returns incorrect results (gray colored rows in Table \ref{tab:SimulationStudy2-diff}) when the intramodule connectivity is low.

\begin{table}[ht!]
\caption{Study 2 \& 3. The performance results summarized as average p-values for comparing various groups of networks. For column K, 2 vs. 2 means we compared a network with 2 modules against another network with 2 modules,  $p$ is the intramodule connectivity probability. TII is an acronym for the {\em topological invariant inference} procedure, BL is the baseline supremum test procedure. TII(HKS) is the performance after heat kernel smoothing.}
\label{tab:SimulationStudy2-diff}
\resizebox{\textwidth}{!}{%
\centering
\begin{tabular}{cc|cccc}
    K & p & BL & TII & TII(HKS) \\ \hline
  \multirow{2}{*}{2 vs. 2} & 0.6 & $0.52 \pm 0.30$ & $0.51 \pm 0.32$ & $0.51 \pm 0.33$ \\
   & 0.8 & $0.54 \pm 0.33$ & $0.53 \pm 0.28$ & $0.54 \pm 0.27$ \\ \hline
  \multirow{2}{*}{5 vs. 5} & 0.6 & $0.47 \pm 0.30$ & $0.54 \pm 0.32$ & $0.51 \pm 0.30$ \\
   & 0.8 & $0.51 \pm 0.32$ & $0.55 \pm 0.31$ & $0.53 \pm 0.30$ \\ \hline
  \cellcolor{white} \multirow{2}{*}{2 vs. 5} & \cellcolor{white} 0.6 & $0.22 \pm 0.21$ & $0.16 \pm 0.24$ & $0.19 \pm 0.27$ \\
   & \cellcolor{white} 0.8 & $0.00 \pm 0.01$ & $0.00 \pm 0.00$ & $0.00 \pm 0.00$ \\ \hline
  
  \cellcolor{white} \multirow{2}{*}{3 vs. 6} & \cellcolor{white} 0.2 & \cellcolor{Gray} $0.58 \pm 0.29$ & \cellcolor{Gray}  $0.01 \pm 0.03$ & \cellcolor{Gray} $0.01 \pm 0.02$ \\
   & \cellcolor{white} 0.4 & $0.59 \pm 0.31$ & $0.35 \pm 0.24$ & $0.36 \pm 0.28$ \\ \hline
  
  \multirow{2}{*}{3 vs. 3} & 0.2 & $0.51 \pm 0.31$ & $0.59 \pm 0.31$ & $0.58 \pm 0.27$ \\
   & 0.4 & $0.52 \pm 0.30$ & $0.57 \pm 0.30$ & $0.53 \pm 0.30$ \\ \hline
  \multirow{2}{*}{6 vs. 6} & 0.2 & $0.54 \pm 0.30$ & $0.59 \pm 0.31$ & $0.58 \pm 0.27$ \\
   & 0.4 & $0.54 \pm 0.30$ & $0.59 \pm 0.34$ & $0.51 \pm 0.32$ \\
 \hline \hline
\end{tabular}
}
\end{table}

{\bf Topological equivalence.}  Networks with similar modularity structures are compared, i.e. 2 vs. 2, 3 vs. 3, 5 vs. 5 and 6 vs. 6. These comparisons were varied based on the intramodule connectivity probability of $\pi = \{0.2, 0.4, 0.6, 0.8\}$. The topology of each group will be similar and we expect the topologically invariant test procedure to detect the network similarity. The results of these comparisons are summarized in Table \ref{tab:SimulationStudy2-diff}. Both the baseline method and the topologically invariant procedure have accurate and comparable performance results when distinguishing networks of the same topologies.

\subsection{Study III: Heat kernel smoothing improves performance}
We assessed and compared the performance of the baseline method and the topological invariant inference procedure after smoothing the networks. In each case the, full set of basis for the network was used in estimating the heat kernel coefficients. The level of smoothing was moderated by a bandwidth of 0.05, which was chosen empirically. We note that in application, the chosen bandwidth will depend on specific level of smoothing desired and the peculiarities of each application. The performance results, summarized in Table \ref{tab:SimulationStudy2-diff} as BL, TII, and TII (HKS) for the baseline and the invariant inference procedure. 
The topological invariant inference procedure has more significant performance gain compared to the baseline method.

\subsection{Study IV: Comparing TII against common loss functions}
We will demonstrate in this study that the topological invariant inference with the topological loss outperforms popular loss functions in the literature. The topological loss was compared to the $l_1, l_2, l_\infty$ norms, adapted as loss functions. The bottleneck, Gromov-Hausdorff (GH) \cite{cohen2007stability} distances were also used as test statistics.

{\bf Topological difference.} We compared networks with different modularity structures, 2 vs. 5 for varying number of nodes and intramodule connectivity probabilities. The topology of each group will be different and we expect each loss function to detect the network difference. The results of these comparisons are summarized in Table \ref{tab:SimulationStudy4-other-loss}. The {\em topologically invariant inference} procedure with the topological loss outperforms all the other loss functions in all but few of the cases, when distinguishing networks of different topologies, especially with low intramodule connectivity probability. The bottle neck and GH distances were the worst performing methods and tend to return incorrect results (gray colored cells in Table \ref{tab:SimulationStudy4-other-loss}) in most of the case.

\begin{table}[ht!]
\caption{Study IV. The performance results summarized as average p-values for comparing various groups of networks. The column v contains the number of nodes used to construct the network in each group.  For column K, 2 vs. 2 means we compared a network with 2 modules against another network with 2 modules,  $p$ is the intramodule connectivity probability. $L_1, L_2, L_\infty$ are loss functions based on the $l_1, l_2, l_\infty$ norms respectively. BN is the bottle neck distance, and GH is the Gromov-Hausdor distance. TII is an acronym for the {\em topological invariant inference} procedure with topological loss. For networks with the different modularity structures, smaller values are preferred. For networks with the same modularity structure, larger values are preferred.}
\label{tab:SimulationStudy4-other-loss}
\renewcommand{\arraystretch}{1.0}
\centering
\resizebox{\textwidth}{!}{%
\begin{tabular}{ccc|ccccccc}
    v & K & p & $L_1$ & $L_2$ & $L_\infty$ & BN & GH & TII \\ \hline
  \multirow{6}{*}{10} & \multirow{2}{*}{2 vs. 2} & 0.6 & $0.56 \pm 0.28$ & $0.54 \pm 0.28$ & $0.49 \pm 0.25$ & $0.15 \pm 0.05$ & $0.43 \pm 0.46$ & $0.58 \pm 0.27$\\ 
  & & 0.8 & $0.39 \pm 0.24$ & $0.39 \pm 0.26$ & $0.34 \pm 0.25$ & $0.13 \pm 0.08$ & $0.31 \pm 0.41$ & $0.45 \pm 0.25$\\ 
  & \multirow{2}{*}{5 vs. 5} & 0.6 & $0.50 \pm 0.28$ & $0.50 \pm 0.29$ & $0.50 \pm 0.27$ & $0.15 \pm 0.09$ & $0.39 \pm 0.43$ & $0.52 \pm 0.29$ \\
  & & 0.8 & $0.55 \pm 0.31$ & $0.54 \pm 0.32$ & $0.52 \pm 0.32$ & $0.13 \pm 0.16$ & $0.33 \pm 0.41$ & $0.52 \pm 0.32$\\ 

  & \cellcolor{white} \multirow{2}{*}{2 vs. 5} & \cellcolor{white}  0.6 & $0.09 \pm 0.12$ & $0.09 \pm 0.11$ & $0.10 \pm 0.13$ & $0.10 \pm 0.10$ & $0.16 \pm 0.33$ & $0.07 \pm 0.08$ \\
  & & 0.8 & $0.00 \pm 0.00$ & $0.00 \pm 0.00$ & $0.00 \pm 0.00$ & $0.07 \pm 0.10$ & $0.14 \pm 0.30$ & $0.00 \pm 0.00$\\ 
 \hline \hline
 
  \multirow{6}{*}{20} & \multirow{2}{*}{2 vs. 2} & 0.6 & $0.54 \pm 0.29$ & $0.54 \pm 0.27$  & $0.55 \pm 0.27$ & \cellcolor{Gray} $0.06 \pm 0.10$ & $0.53 \pm 0.49$ & $0.58 \pm 0.23$\\ 
  & & 0.8 & $0.46 \pm 0.28$ & $0.44 \pm 0.28$ & $0.45 \pm 0.28$ & \cellcolor{Gray} $0.03 \pm 0.12$ & $0.25 \pm 0.42$ & $0.46 \pm 0.29$\\ 
  & \multirow{2}{*}{5 vs. 5} & 0.6 & $0.52 \pm 0.31$ & $0.50 \pm 0.30$ & $0.52 \pm 0.30$ & \cellcolor{Gray} $0.03 \pm 0.12$ & $0.24 \pm 0.43$ & $0.53 \pm 0.31$\\
  & & 0.8 & $0.54 \pm 0.29$ & $0.53 \pm 0.30$ & $0.52 \pm 0.28$ & \cellcolor{Gray} $0.00 \pm 0.00$ & \cellcolor{Gray} $0.02 \pm 0.14$ & $0.56 \pm 0.30$ \\ 

 &  \cellcolor{white} \multirow{2}{*}{2 vs. 5} &  \cellcolor{white} 0.6 & $0.01 \pm 0.02$ & $0.01 \pm 0.02$ & $0.02 \pm 0.05$ & $0.05 \pm 0.13$ & \cellcolor{Gray} $0.15 \pm 0.34$ & $0.00 \pm 0.02$ \\
  & & 0.8 & $0.00 \pm 0.00$ & $0.00 \pm 0.00$ & $0.00 \pm 0.00$ & $0.00 \pm 0.03$ & $0.02 \pm 0.14$ & $0.00 \pm 0.00$\\ 
 \hline \hline
 
  \multirow{6}{*}{30} & \multirow{2}{*}{2 vs. 2} & 0.6 & $0.53 \pm 0.31$ & $0.52 \pm 0.31$ & $0.54 \pm 0.29$ & \cellcolor{Gray} $0.00 \pm 0.00$ & \cellcolor{Gray} $0.06 \pm 0.24$ & $0.53 \pm 0.29$ \\ 
  & & 0.8 & $0.52 \pm 0.32$ & $0.54 \pm 0.34$ & $0.51 \pm 0.33$ & \cellcolor{Gray} $0.00 \pm 0.00$ & $0.10 \pm 0.32$ & $0.60 \pm 0.34$ \\ 
  & \multirow{2}{*}{5 vs. 5} & 0.6 & $0.56 \pm 0.28$ & $0.59 \pm 0.29$ & $0.59 \pm 0.29$ & \cellcolor{Gray} $0.00 \pm 0.00$ & \cellcolor{Gray} $0.03 \pm 0.18$ & $0.61 \pm 0.31$\\
  & & 0.8 & $ 0.51 \pm 0.29$ & $0.50 \pm 0.28$ & $0.49 \pm 0.29$ & \cellcolor{Gray} $0.00 \pm 0.00$ & \cellcolor{Gray} $0.00 \pm 0.00$ & $0.54 \pm 0.27$\\ 

 &  \cellcolor{white} \multirow{2}{*}{2 vs. 5} &  \cellcolor{white} 0.6 & $0.00 \pm 0.01$ & $0.01 \pm 0.02$ & $0.01 \pm 0.04$ & $0.00\pm 0.00$ & $0.00 \pm 0.00$ & $0.00 \pm 0.01$ \\
  & & 0.8 & $0.00 \pm 0.00$ & $0.00 \pm 0.00$ & $0.00 \pm 0.00$ & $0.00 \pm 0.00$ & $0.00 \pm 0.00$ & $0.00 \pm 0.00$\\ 
 \hline \hline
 
\end{tabular}
} 
\end{table}

{\bf Topological equivalence.}  Networks with similar modularity structures are compared, i.e. 2 vs. 2, 5 vs. 5. These comparisons were varied based on the intramodule connectivity probability of $\pi = \{0.6, 0.8\}$. The topology of each group will be similar and we expect each loss function to detect the network similarity. The results of these comparisons are summarized in Table \ref{tab:SimulationStudy4-other-loss}. TII with topological loss performs better than all other loss functions in majority of the case. The bottle neck and GH distances were the worst performing methods and tend to return incorrect results (gray colored cells in Table \ref{tab:SimulationStudy4-other-loss}) in some of the case.

\section{Application to Functional Brain Networks}
\label{sec:app}
The proposed framework is applied to real data with an edge-domain analysis of human resting-state fMRI. The dataset is provided as parcel-wise functional connectivity networks (see figure~\ref{fig:mta-raw-non-use}). The Heat Kernel Expansion regularizes the connectivity signal on the graph while preserving topology. This application illustrates the practical value of edge-domain processing for functional connectomics.
\begin{figure*}[ht]
     \centering
     \begin{subfigure}[t]{0.325\linewidth}
         \centering
         \includegraphics[width=\textwidth]{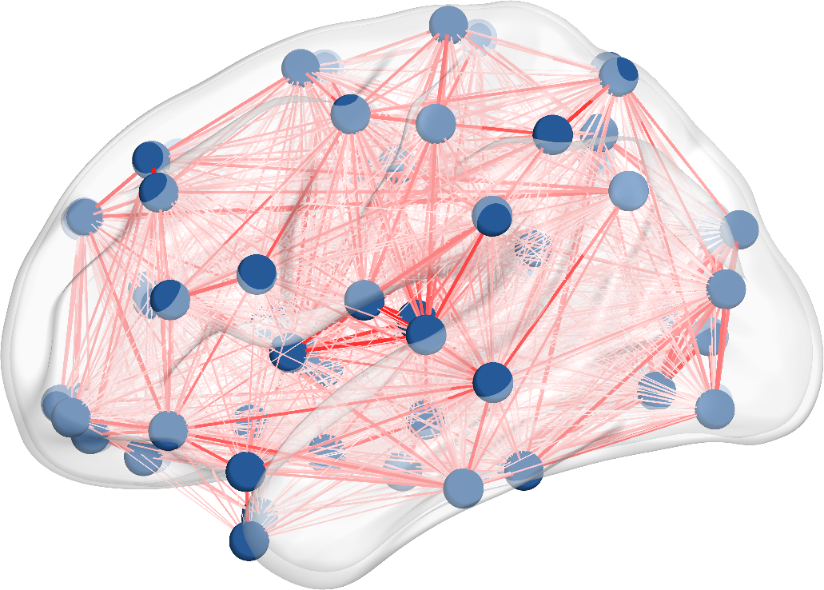}
     \end{subfigure}
     \hfill
     \begin{subfigure}[t]{0.325\linewidth}
         \centering
         \includegraphics[width=0.75\textwidth,height=0.75\textwidth]{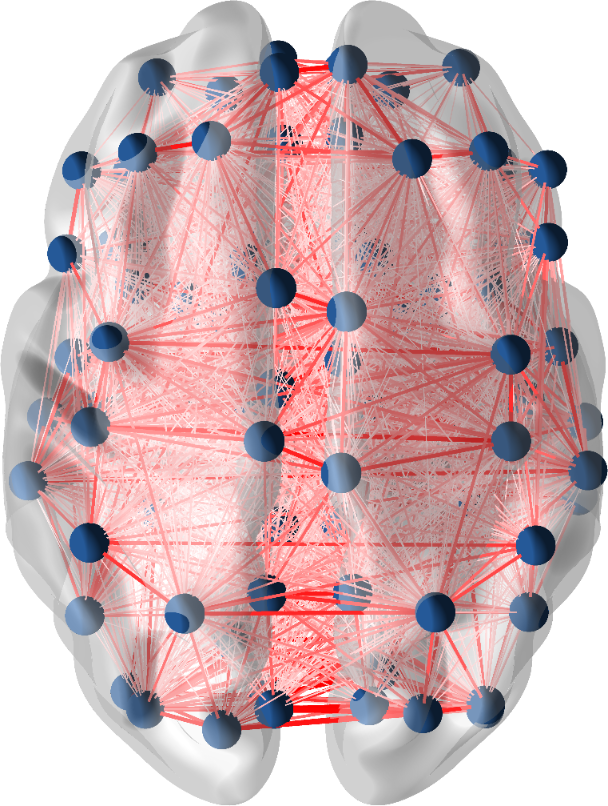}
     \end{subfigure}
     \hfill
     \begin{subfigure}[t]{0.325\linewidth}
         \centering
         \includegraphics[width=\textwidth]{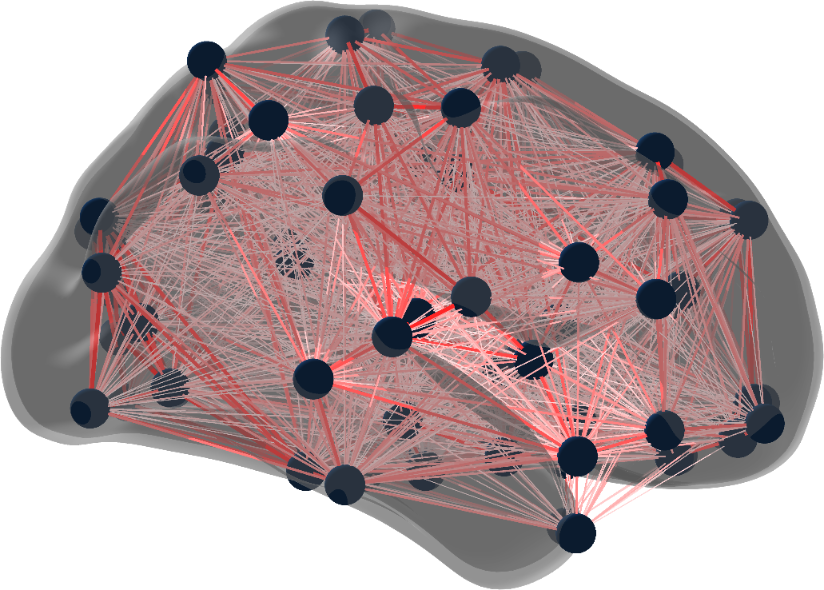}
     \end{subfigure}
     
    \par\vspace{0.5cm}
    
     \centering
     \begin{subfigure}[t]{0.325\linewidth}
         \centering
         \includegraphics[width=\textwidth]{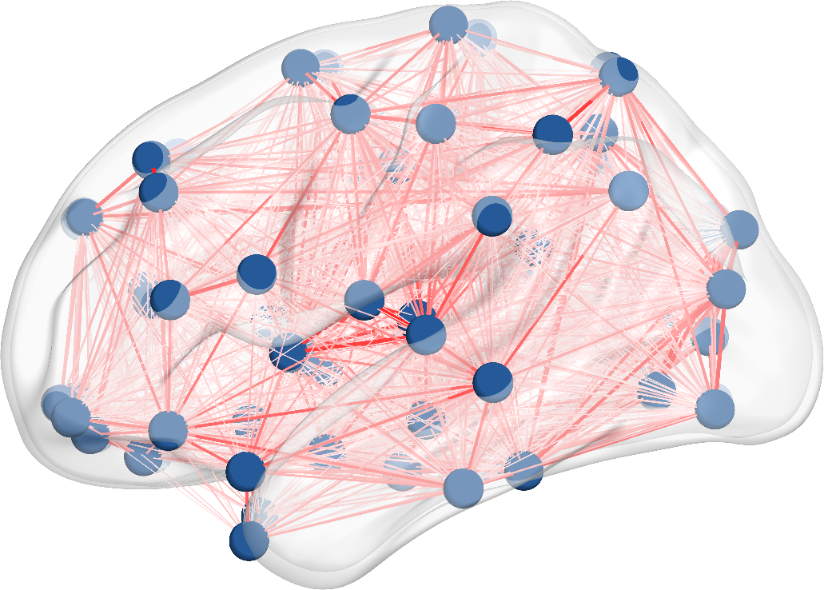}
     \end{subfigure}
     \hfill
     \begin{subfigure}[t]{0.325\linewidth}
         \centering
         \includegraphics[width=0.75\textwidth,height=0.75\textwidth]{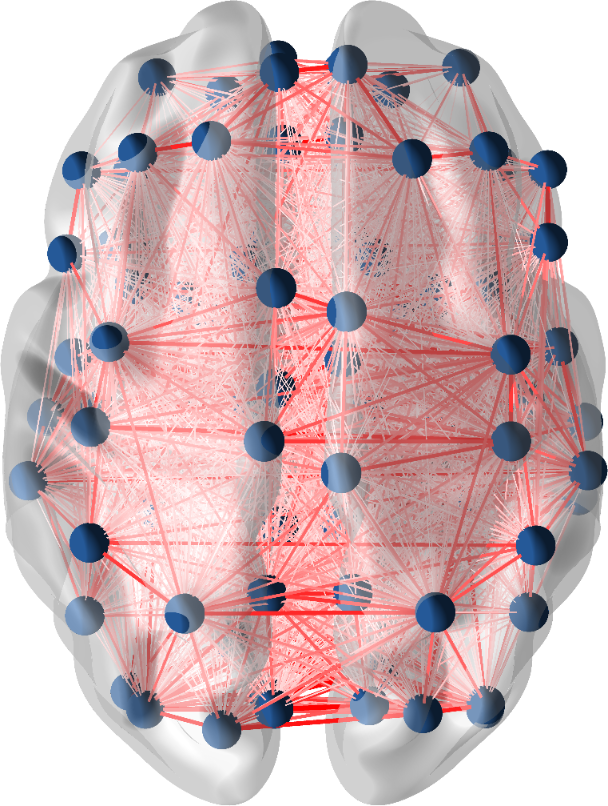}
     \end{subfigure}
     \hfill
     \begin{subfigure}[t]{0.325\linewidth}
         \centering
         \includegraphics[width=\textwidth]{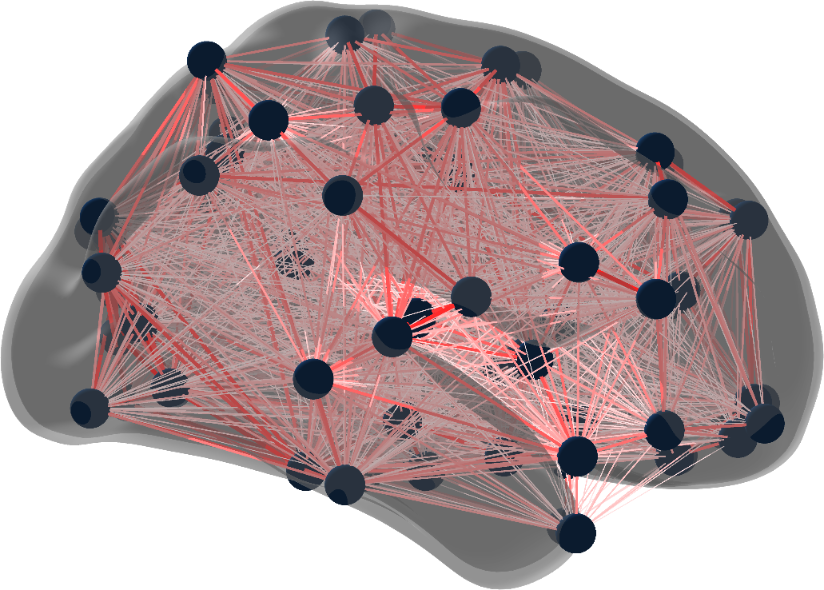}
     \end{subfigure}
      \begin{subfigure}[t]{1\linewidth}
         \centering
         \includegraphics[width=0.75\textwidth]{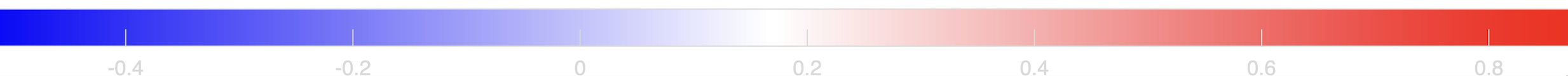}
     \end{subfigure}
    \caption{Top: Group-mean functional connectivity of users. Relatively few negative connections were observed. Bottom: Group-mean functional connectivity of non-users. Relatively few negative connections were observed. The dense representation of the connectivities gives no visibly discernible pattern and might suggest topoligical equivalence.}
    \label{fig:mta-raw-non-use}
\end{figure*}

\subsection{Data and Preprocessing}
The data analyzed is the Addiction Connectome Preprocessed Initiative (ACPI) distribution of the \emph{Multimodal Treatment of ADHD (MTA)} resting-state fMRI resource released through FCP/INDI. This is a derivative, preprocessed release: raw scans were curated by the project and distributed both as standardized 4D volumes and as parcellated region–wise time series. The parcellated time series in the Automated Anatomical Labeling (AAL) atlas was used for this analysis \cite{tzourio2002automated}. The data is in two groups based on marijuana use status, group that regularly use cannabis (users) and another group that did not regularly use cannabis (non-users). 

Preprocessing in ACPI is provided as a small, explicit factorial set of pipelines that differ only in three choices made after standard motion correction and normalization to MNI space. The first factor is the nonlinear spatial normalization tool (\textit{ANTS} or \textit{FNIRT}); the second is whether high–motion volumes were censored; the third is whether global signal regression was applied. Each subject therefore has up to eight variants, all parcellated to the same AAL ROI set. We selected a single variant \emph{a priori} for all subjects to avoid mixing preprocessing choices across groups. For each subject we transformed the ROI timeseries by removing ROI means, and computing the Pearson correlation between every ROI pair to obtain a symmetric AAL connectivity matrix. Correlations were Fisher $z$–transformed prior to group averaging and back–transformed for reporting and visualization. No additional denoising, filtering, spatial smoothing, or regression was performed beyond the ACPI pipeline settings. This design keeps the preprocessing provenance explicit and comparable across subjects while providing standard parcel–level functional networks for the subsequent edge–domain analyses.

\begin{figure*}[ht]
     \centering
     \begin{subfigure}[t]{0.325\linewidth}
         \centering
         \includegraphics[width=\textwidth]{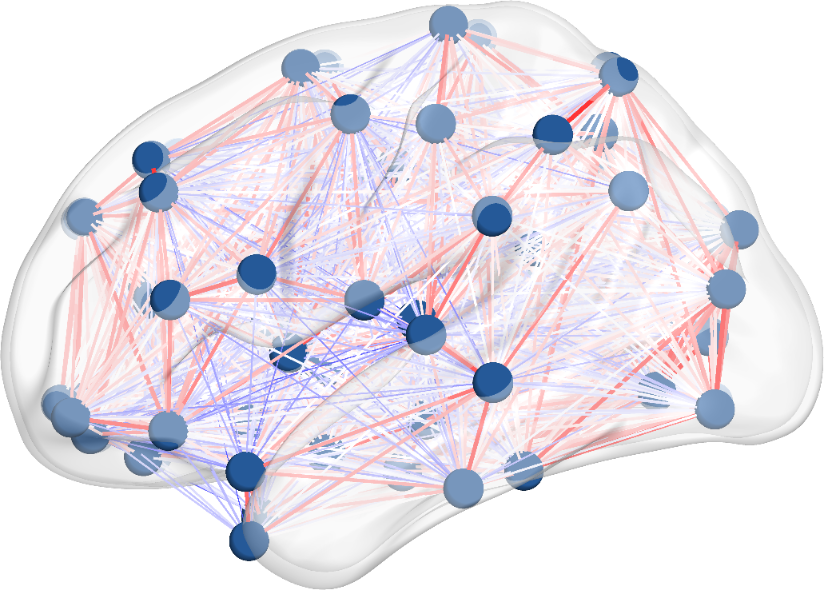}
     \end{subfigure}
     \hfill
     \begin{subfigure}[t]{0.325\linewidth}
         \centering
         \includegraphics[width=0.75\textwidth,height=0.75\textwidth]{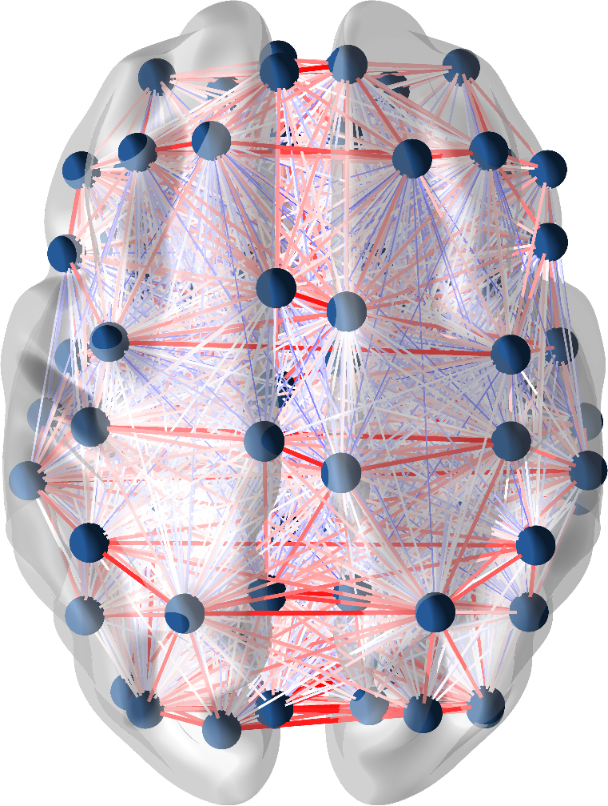}
     \end{subfigure}
     \hfill
     \begin{subfigure}[t]{0.325\linewidth}
         \centering
         \includegraphics[width=\textwidth]{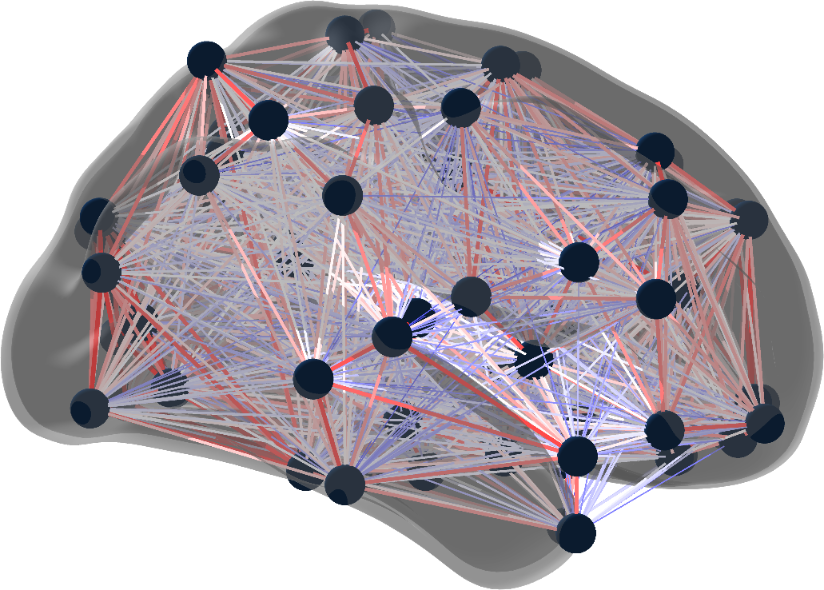}
     \end{subfigure}
     
    \par\vspace{0.5cm}
    
     \centering
     \begin{subfigure}[t]{0.325\linewidth}
         \centering
         \includegraphics[width=\textwidth]{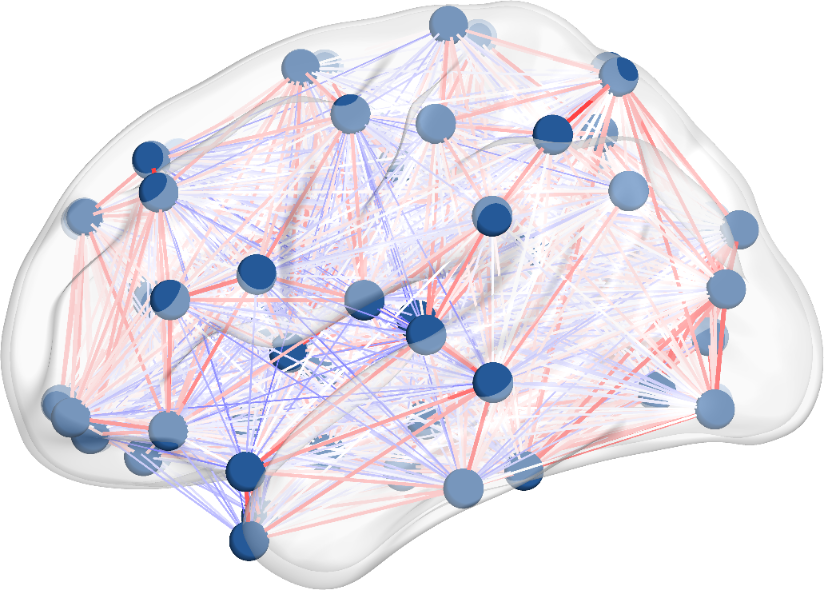}
     \end{subfigure}
     \hfill
     \begin{subfigure}[t]{0.325\linewidth}
         \centering
         \includegraphics[width=0.75\textwidth,height=0.75\textwidth]{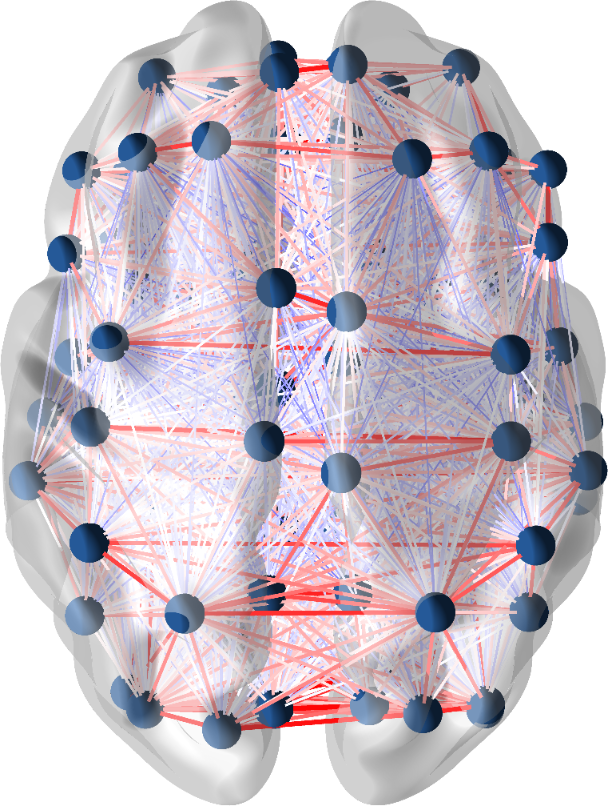}
     \end{subfigure}
     \hfill
     \begin{subfigure}[t]{0.325\linewidth}
         \centering
         \includegraphics[width=\textwidth]{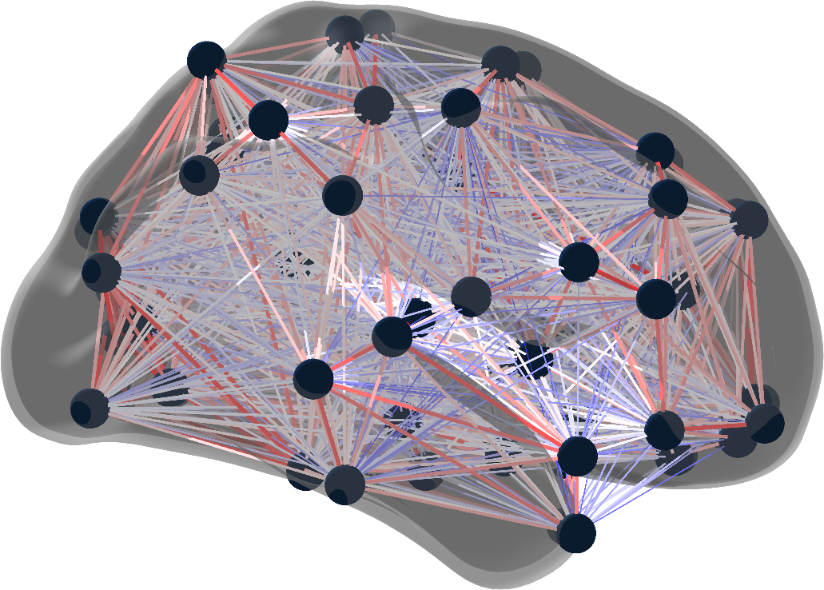}
     \end{subfigure}
      \begin{subfigure}[t]{1\linewidth}
         \centering
         \includegraphics[width=0.75\textwidth]{cb_use1.png}
     \end{subfigure}
    \caption{Top: The functionality connectivity after heat kernel expansion with bandwidth $t=0.01$ of the mean connectivity for users. Observe that spurious isolated connections are attenuated (connectivity reduced) while coherent bundles are enhanced. Bottom: The functionality connectivity after heat kernel expansion with bandwidth $t=0.05$ of the mean connectivity for non-users. Observe that spurious isolated connections are attenuated (connectivity reduced) while coherent bundles are enhanced.}
    \label{fig:mta-smooth-non-use2}
\end{figure*}

\subsection{Statistical Results}
The method is applied to determine whether there is topological dissimilarity between the functional brain networks of users and non-users. There are 62 users and 63 non-users. The comparison is made at three levels of sample sizes: 5 samples, 10 samples and 25 samples. The sample sizes are equal within each group. To test for the existence of topological dissimilarity for each of these samples, we generated 1M permutations with 50 independent repetitions to access the standard error. For each permutation, the test statistic \eqref{eqn:total-top-loss-test-statistic} is recorded. Figure~\ref{fig:convg-pvals-5samples} shows the p-value convergence plot and the distribution of the test statistics for 5 samples. The p-value obtained is $0.0.0885 \pm 0.0003$ indicating the presence of topological dissimilarity. 
\begin{figure}[ht]
\centering
 \begin{subfigure}[b]{0.45\textwidth}
     \centering
     \includegraphics[width=\textwidth]{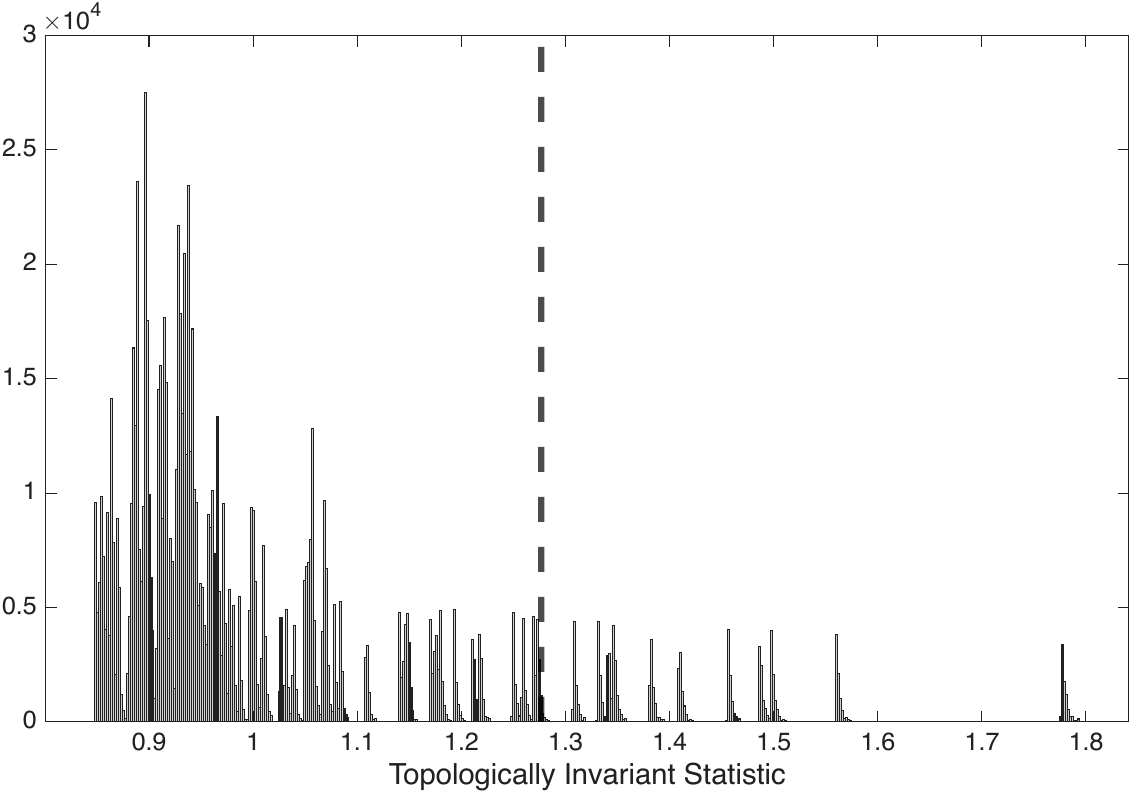}
 \end{subfigure}
 \hspace{1em}
 \begin{subfigure}[b]{0.45\textwidth}
     \centering
     \includegraphics[width=\textwidth]{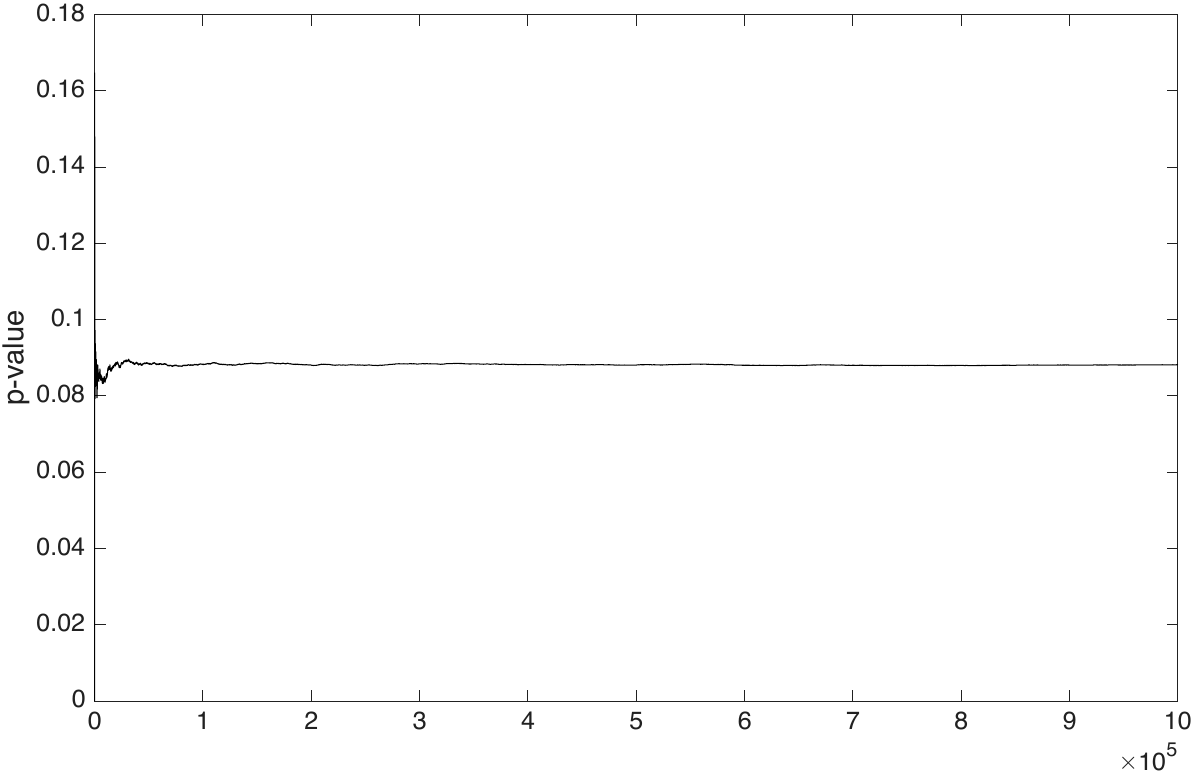}
\end{subfigure}
\caption{Left: The empirical distribution of the test statistic for the 5 samples comparison. The dashed vertical line indicates the test statistic. Right: The comparison p-value after 1M permutations with 50 independent repetitions. The p-value obtained is $0.0885 \pm 0.0003$, and the smooth horizontal line indicates convergence. }
\label{fig:convg-pvals-5samples}
\end{figure}

Similarly, Figure~\ref{fig:convg-pvals-10samples} and Figure~\ref{fig:convg-pvals-25samples} shows the p-value convergence plot and the distribution of the test statistics for 10 samples and 25 samples respectively. The p-value obtained is $0.0210 \pm 0.0002$ and $0.0756 \pm 0.0003$ for the 10 samples and 25 samples respectively indicating the test procedure successfully detected topological inequivalence in both samples. 

\begin{figure}[ht]
\centering
 \begin{subfigure}[b]{0.45\textwidth}
     \centering
     \includegraphics[width=\textwidth]{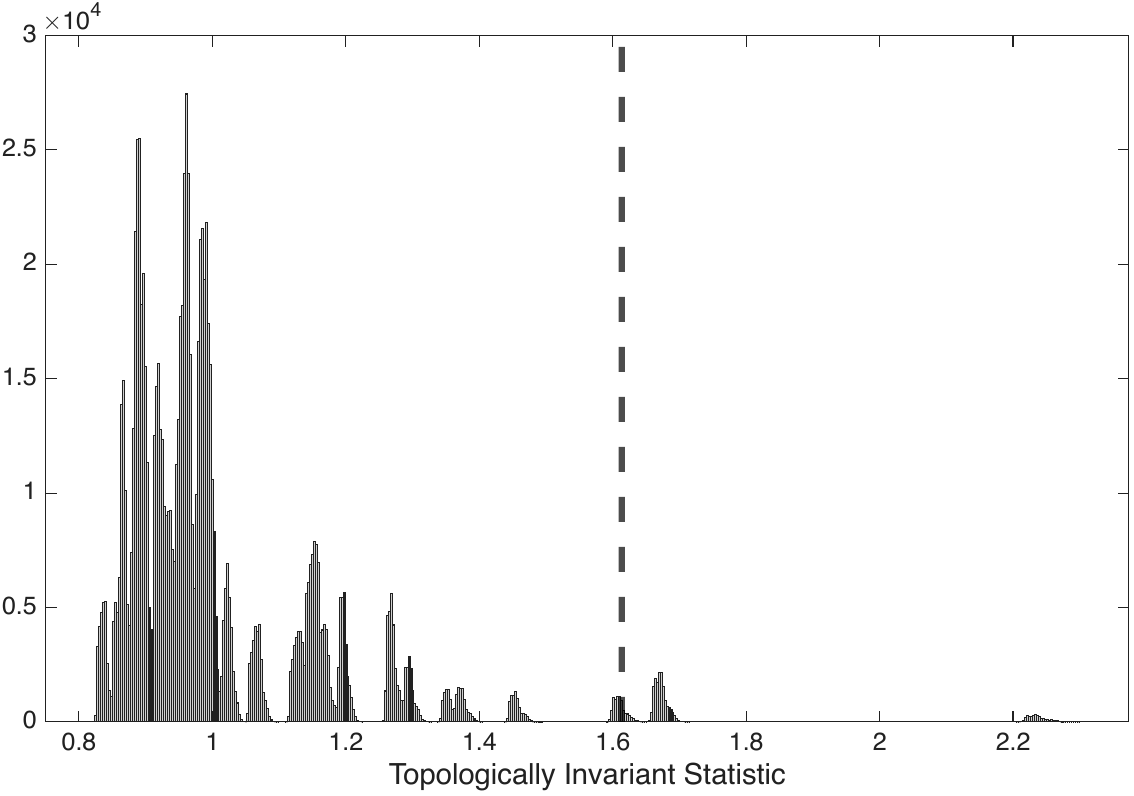}
 \end{subfigure}
 \hspace{1em}
 \begin{subfigure}[b]{0.45\textwidth}
     \centering
     \includegraphics[width=\textwidth]{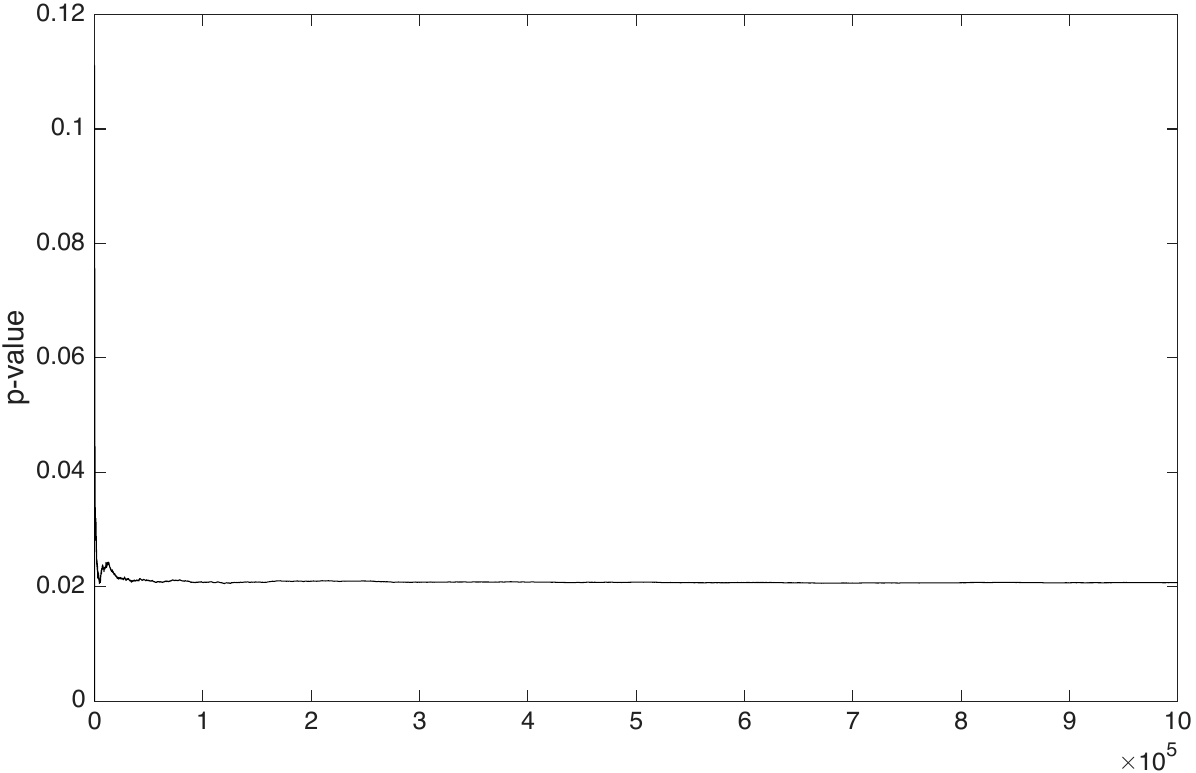}
\end{subfigure}
\caption{Left: The empirical distribution of the test statistic for the 10 samples comparison. The dashed vertical line indicates the test statistic. Right: The comparison p-value after 1M permutations with 50 independent repetitions. The p-value obtained is $0.0210 \pm 0.0002$, and the smooth horizontal line indicates convergence. }
\label{fig:convg-pvals-10samples}
\end{figure}
\begin{figure}[ht]
\centering
 \begin{subfigure}[b]{0.45\textwidth}
     \centering
     \includegraphics[width=\textwidth]{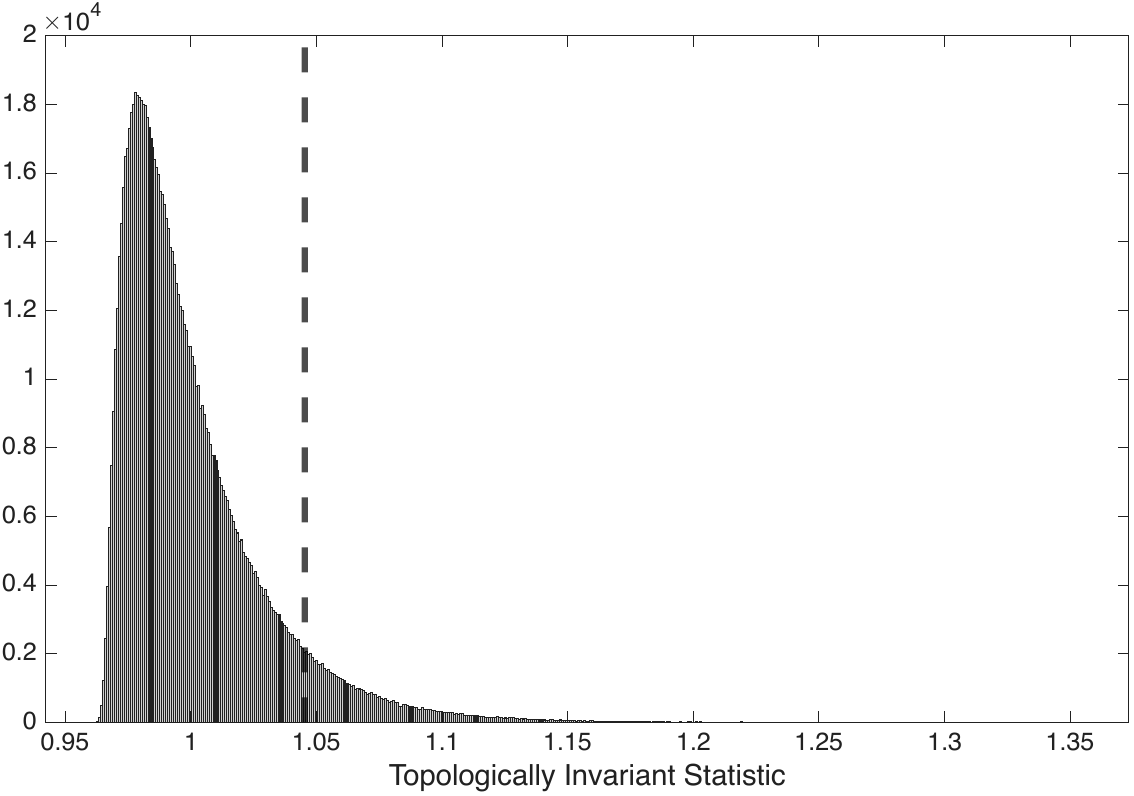}
 \end{subfigure}
 \hspace{1em}
 \begin{subfigure}[b]{0.45\textwidth}
     \centering
     \includegraphics[width=\textwidth]{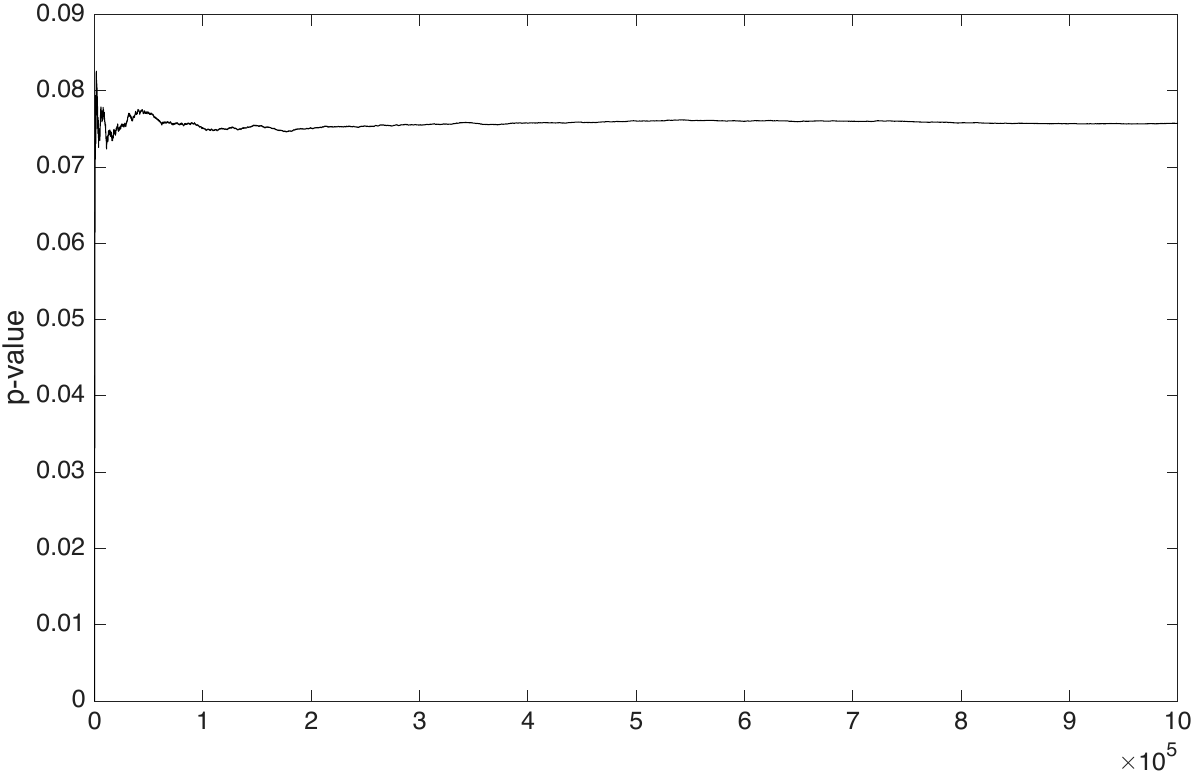}
\end{subfigure}
\caption{Left: The empirical distribution of the test statistic for the 5 samples comparison. The dashed vertical line indicates the test statistic. Right: The comparison p-value after 1M permutations with 50 independent repetitions. The p-value obtained is $0.0756 \pm 0.0003$, and the smooth horizontal line indicates convergence. }
\label{fig:convg-pvals-25samples}
\end{figure}

\section{Discussion \& Conclusion}
\label{sec:disc}

This work introduces a topologically invariant permutation test for comparing brain networks through persistent homology, combined with heat kernel expansion on the Hodge Laplacian for reducing topological variability. The invariance enables a resampling scheme that preserves individual network topologies while generating valid null distributions for testing topological equivalence between groups. The test statistic based on $2$-Wasserstein distances provides closed-form computation without expensive optimal transport solvers, making the approach computationally feasible for neuroimaging applications.

The heat kernel expansion on the Hodge Laplacian offers principled smoothing of connectivity matrices with theoretical guarantees on variance reduction. Unlike ad-hoc filtering procedures, the bandwidth parameter $t$ and expansion degree $d$ provide intuitive control over smoothing intensity while maintaining topological features through Hilbert space projection. Our simulations demonstrate that this smoothing enhances the statistical power of the topologically invariant test, particularly for networks with low connectivity density or subtle modular structure. The application to cannabis use in the ADHD population reveals significant topological differences that manifest as altered modular organization and cycle structure, demonstrating the method's sensitivity to biologically meaningful network alterations. By permuting birth and death values separately rather than entire networks, our approach directly tests topological equivalence while accommodating individual differences in connectivity patterns. The superior performance compared to baseline two-sample tests and alternative distance metrics (bottleneck, Gromov-Hausdorff, $\ell_p$ norms) validates this choice across diverse simulation scenarios.

Several limitations warrant consideration. First, the detected topological differences in cannabis users do not establish causality, as pre-existing network differences may have influenced substance use patterns. Longitudinal studies with pre- and post-exposure measurements would be needed to address this question. Second, the choice of atlas parcellation influences network topology by defining node locations and sizes. While we expect qualitative findings to be robust across reasonable parcellations, systematic investigation of parcellation effects remains important. Third, heat kernel bandwidth selection currently relies on empirical tuning, though the theoretical framework provides guidance through the variance reduction guarantees. Developing principled, data-driven bandwidth selection methods would enhance the framework's accessibility and reproducibility.

Future extensions could incorporate higher-order homology ($\beta_2$ and beyond) to characterize cavities and higher-dimensional voids, though the biological interpretation of these features in brain networks requires further investigation. The framework could also be extended to dynamic networks, enabling characterization of how topology evolves during cognitive tasks, learning, or disease progression. Time-varying persistent homology offers a promising direction for such extensions.

In conclusion, this work provides a rigorous mathematical framework for topological analysis of brain networks that combines persistent homology, Hodge Laplacian theory, and topologically invariant statistical inference. The detection of significant topological differences in cannabis users' brain networks illustrates the practical value of this approach for understanding how environmental factors and behaviors alter neural organization. Topological methods will play an increasingly important role in characterizing brain function and dysfunction, offering insights that complement traditional graph-theoretic approaches while respecting the intrinsic geometry of neural connectivity.

\bibliographystyle{plain}
\bibliography{references}

\end{document}